\documentclass[leqno, 11pt, twoside, letterpaper]{amsart}

\numberwithin{equation}{section}
\usepackage{geometry}
\allowdisplaybreaks

\usepackage{enumerate}

\newcommand{\dbra}[1]{[\kern-0.15em[ #1 ]\kern-0.15em]}
\newcommand{\dbraco}[1]{[\kern-0.15em[ #1 [\kern-0.15em[}
\newcommand{\dbraoc}[1]{]\kern-0.15em] #1 ]\kern-0.15em]}
\newcommand{\dbraoo}[1]{]\kern-0.15em] #1 [\kern-0.15em[}

\usepackage{mathabx}
\usepackage{amsmath, amsthm}
\usepackage{amsfonts,amssymb,dsfont}
\usepackage{graphicx}
\usepackage{pgf}
\usepackage{nicefrac,booktabs,mathrsfs}
\usepackage{bm}  
\usepackage{stmaryrd}
\usepackage[backgroundcolor=white,bordercolor=orange]{todonotes}
\usepackage{eurosym}

\newcommand{\cA}{{\mathcal A}}  

\newcommand{\cB}{{\mathcal B}}
\newcommand{\cO}{{\mathcal O}}

\newcommand{\ccF}{{\mathscr F}}\newcommand{\cF}{{\mathcal F}}
\newcommand{\ccG}{{\mathscr G}}\newcommand{\cG}{{\mathcal G}}

\newcommand{\cP}{{\mathcal P}}
\newcommand{\cE}{{\mathcal E}}

\newcommand{\Ind}{{\mathds 1}}
\newcommand{\ind}[1]{\Ind_{\{#1\}}}

\newcommand{\R}{\mathbb{R}}

\newcommand{\FF}{\mathbb{F}}
\newcommand{\bbF}{\mathbb{F}}  

\newcommand{\bbG}{\mathbb{G}}
\newcommand{\N}{\mathbb{N}}

\newcommand{\E}{\mathbb{E}}

\newcommand{\Q}{\mathbb{Q}}
\newcommand{\T}{\mathbb{T}}
\newcommand{\prob}{\mathbb{Q}} 

\newtheorem{theorem}{Theorem}[section]
\newtheorem{corollary}[theorem]{Corollary}      
\newtheorem{lemma}[theorem]{Lemma}              
\newtheorem{proposition}[theorem]{Proposition}  

\theoremstyle{definition}
\newtheorem{example}[theorem]{Example}

\newtheorem{remark}[theorem]{Remark}
\newtheorem{assumption}[theorem]{Assumption}

\usepackage[pdftex,colorlinks,urlcolor=red,citecolor=blue,linkcolor=red]{hyperref}

\renewcommand{\cF}{\ccF}
\renewcommand{\cG}{\ccG}

\begin{document}

\title{General Dynamic Term Structures under Default Risk}

\author {Claudio Fontana \and Thorsten Schmidt}
\address{Laboratoire de Probabilit\'es et Mod\`eles Al\'eatoires, Universit\'e Paris Diderot, avenue de France, 75205, Paris, France.}
\email{fontana@math.univ-paris-diderot.fr}
\address{Freiburg University, Dep. of Mathematical Stochastics, Eckerstr. 1, 79104 Freiburg, Germany  \hspace*{\fill}\linebreak
Freiburg Research Institute of Advanced Studies, Alberstr. 29, 79104 Freiburg, Germany \hspace*{\fill}\linebreak
University of Strasbourg Institute for Advanced Study, 5 all\'e du G\'en\'eral Rouvillois, 67083 Strasbourg.}
\email{thorsten.schmidt@stochastik.uni-freiburg.de}
\thanks{The authors are thankful to Anna Aksamit, Josef Teichmann, and the participants of the FWZ-Seminar for valuable discussions on the topic of the present paper and to two anonymous referees for valuable suggestions that helped improve the paper.}
\date{\today}


\begin{abstract} 
We consider the problem of modelling the term structure of defaultable bonds, under minimal assumptions on the default time. In particular, we do not assume the existence of a default intensity and we therefore allow for the possibility of default at predictable times. 
It turns out that this requires the introduction of an additional term in the forward rate approach by Heath, Jarrow and Morton (1992). This term is driven by a random measure encoding information about those times where default can happen with positive probability. 
In this framework, we  derive necessary and sufficient conditions for a reference probability measure to be a local martingale measure for the large financial market of credit risky bonds, also considering general recovery schemes.
\end{abstract}

\maketitle

\section{Introduction} 

The study of  the evolution of the term structure of bond prices in the presence of default risk typically starts from a forward rate model similar to the classical approach of Heath, Jarrow and Morton (HJM) in \cite{HJM}. In this approach,  bond prices are assumed to be absolutely continuous with respect to the lifetime of the bond (maturity). This assumption is typically justified by the argument that, in practice, only a finite number of bonds are liquidly traded and the full term structure is obtained by interpolation, hence is smooth. 

In markets with default risk, however, discontinuities are the rule rather than the exception: the seminal work of Merton \cite{Merton1974} clearly shows such a behavior, as do many other {\em structural models} (see, e.g., \cite{BelangerShreveWong2004,Geske1977,GeskeJohnson84}).
A default event usually occurs in correspondence of a missed payment by a corporate or sovereign entity and, in many cases, the payment dates are publicly known in advance. 
The missed coupon payments by Argentina on a notional of \$29 billion (on July 30, 2014; see \cite{ISDAArgentina2014}) and by Greece on a notional of \euro1.5 billion (on June 30, 2015; see \cite{NYTimes2015}) are prime examples of credit events occurring at predetermined payment dates.
It is therefore natural to expect the term structure of default risky bonds to exhibit discontinuities in correspondence of such payment dates.\footnote{As an illustrative example, the timeline of the payment dates on Greece's debt is publicly available and daily updated at \texttt{http://graphics.wsj.com/greece-debt-timeline}.}

On the other side, {\em reduced-form models} (see  \cite{ArtznerDelbaen95,DuffieSchroderSkiadas96,EJY,JarrowLandoTurnbull97,Lando94} for some of the first works in this direction) are less ambitious about the precise mechanism leading to default and neglect this phenomenon.  
Reduced-form models generally assume the existence of a {\em default intensity}, thus implying that the probability of the default event occurring at any predictable time vanishes. 
Accordingly, reduced-form HJM-type models for defaultable term structures typically postulate that, prior to default, bond prices are absolutely continuous with respect to maturity, i.e., under the assumption of zero recovery, credit risky bond prices $P(t,T)$ are described by
\begin{equation} \label{HJM:intro}
P(t,T) = \ind{\tau > t} \exp\bigg( - \int_t^T f(t,u) du \bigg),
\end{equation}
with $\tau$ denoting the random default time and $(f(t,T))_{0\leq t\leq T}$ an instantaneous forward rate.
This approach has been studied in numerous works and up to a great level of generality, beginning with the first works \cite{DuffieSingleton99,JarrowTurnbull1995,Schoenbucher:CRDerivatives,Schoenbucher98}  and extended in various directions in \cite{EbGr,EbOz,SchmidtOezkan,TSchmidt_InfiniteFactors} (see  \cite[Chapter 13]{BieleckiRutkowski2002} for an overview of the relevant literature). 

It turns out that, assuming absence of arbitrage, the presence of predictable times at which the default event can occur with strictly positive probability is incompatible with an absolutely continuous term structure of the form \eqref{HJM:intro}. This fact, already pointed out in 1998 in \cite{Schoenbucher98}, has motivated more general approaches such as \cite{BelangerShreveWong2004} and \cite{GehmlichSchmidt2016MF} (see Section~\ref{sec:rel-literature} for an overview of the related literature). 
In particular, in the recent paper \cite{GehmlichSchmidt2016MF}, the classical reduced-form HJM approach is extended by allowing the {\em default compensator} 
to have an absolutely continuous part, corresponding to a default intensity, as well as a discontinuous part with a finite number of jumps. The presence of jumps allows the default event to occur with strictly positive probability at the predictable jump times, which in \cite{GehmlichSchmidt2016MF} are assumed to be revealed in advance in the market. In this context, in order to exclude arbitrage possibilities, the term structure equation \eqref{HJM:intro} has to be extended by introducing discontinuities in correspondence of those times.

In the present paper, we introduce a general framework for the modelling of defaultable term structures under minimal assumptions, going significantly beyond the intensity-based approach and generalizing the setting of \cite{GehmlichSchmidt2016MF}.
More specifically, we refrain from making \emph{any} assumption on the default time $\tau$ as well as on the default compensator, allowing in particular the default event to occur with strictly positive probability at predictable times. 
To the best of our knowledge, previous approaches to the modelling of defaultable term structures have always imposed some assumptions on $\tau$.
A natural and general way to represent the term structure of credit risky bonds, also allowing for discontinuities, is to extend \eqref{HJM:intro} to the following specification:
\begin{equation}	\label{HJM:intro2}
P(t,T) = \ind{\tau > t} \exp\bigg( - \int_t^T f(t,u) du - \int_{(t,T]}g(t,u)\mu_t(du)\bigg),
\end{equation}
where $(\mu_t(du))_{t\geq0}$ is a measure-valued process with possibly singular and jump parts and where $(f(t,T))_{0\leq t\leq T}$ and $(g(t,T))_{0\leq t\leq T}$ are two random fields representing instantaneous forward rates.
The additional term $\int_{(t,T]}g(t,u)\mu_t(du)$ can be interpreted as the impact of the information received up to date $t$ about possible ``risky dates'' (i.e., periods at which the default event can occur with strictly positive probability) in the remaining lifetime $(t,T]$ of the bond.  
We refer to Section~\ref{sec:easy_case} for a simple illustration of the term structure specification \eqref{HJM:intro2}.

In this general setting, we obtain necessary and sufficient conditions for a reference probability measure $\Q$ to be a local martingale measure for the infinite-dimensional financial market consisting of all credit risky bonds, thereby ensuring absence of arbitrage in a sense to be precisely specified below. 
Furthermore, we also study the extension of \eqref{HJM:intro2} to the case of a general recovery process over successive credit events. 
In overall terms, the present paper can be regarded as a general HJM-type framework bridging the gap between intensity-based and structural models.
Moreover, despite the level of generality, our HJM-type conditions admit a clear economic interpretation and can be further simplified in several special cases of practical interest, notably in the case where the process $(\mu_t(du))_{t\geq0}$ is generated by an integer-valued random measure.

The importance of allowing for jumps at predictable times is widely acknowledged in the financial literature. This is due to the fact that jumps in prices typically occur in correspondence of macroeconomic announcements (see, e.g., \cite{Johannes}) and macroeconomic announcements are released at scheduled (predictable) dates. 
In this direction, an econometric model allowing for jumps in correspondence of the meeting dates of the Federal Open Market Committee has been developed and tested in \cite{Piazzesi}. It is shown that the introduction of policy-related jumps improves bond pricing  and allows to generate realistic volatility patterns.
This is further substantiated by the analysis of \cite{KimWright} and  \cite{FlemingRemolona}, where it is shown that macroeconomic announcements have a particularly strong effect on maturities of one to five years.
Recent political events like the Brexit and the election of the American president in 2016 have highlighted the significance of discontinuities occurring at scheduled dates in financial markets.
In this perspective, the present paper contributes to the financial literature by providing for the first time a general theory of defaultable term structure modelling in the presence of jumps occurring at scheduled dates.

The paper is structured as follows. Section \ref{sec:tau} contains a description of the setting and the main technical assumptions and presents a general decomposition of the default compensator process. The main results of the paper are presented in Section \ref{sec:main_results}, first in the case of zero recovery at default and then for a general recovery process. Special cases and examples are also discussed, together with relations to the literature (see Section \ref{sec:rel-literature}). Section \ref{sec:proof} contains the proofs of all our results. 

\section{General defaultable term structure models}\label{sec:tau}

\subsection{Setting}	\label{sec:proba_setting}
 
Let $(\Omega, \cA, \prob)$ be a probability space endowed with a filtration $\bbF=(\cF_t)_{0\leq t\leq \T}$ satisfying the usual conditions (i.e., $\bbF$ is right-continuous and, if $A\subseteq B \in \cA$ and $\prob(B)=0$, then $A \in \cF_0$), with $\T<+\infty$ denoting a final time horizon.\footnote{{The infinite time horizon case can be dealt with in a similar way and leads to analogous results, provided that $\mu$ is a random measure on $\R^2_+$ satisfying a suitable version of Assumption \ref{ass:rnd_meas}.}} 
We assume that the filtered probability space $(\Omega, \cA, \bbF, \prob)$ is sufficiently rich to support an $n$-dimensional Brownian motion $W=(W_t)_{0\leq t\leq \T}$ and an optional non-negative random measure $\mu(ds,du)$ on $[0,\T]\times[0,\T]$.
Throughout the paper, the probability measure $\prob$ will represent a reference probability measure. We follow the notation of \cite{JacodShiryaev} and refer to this work for details on stochastic processes which are not laid out here. 

\subsection{The default time}

We consider an abstract economy containing an entity (e.g., a company or a sovereign) which may default at the random {\em default time} $\tau$. The filtration $\bbF$ is meant to represent all information publicly available in the market. The default event is publicly observable,  which implies that the random time $\tau$ is an $\bbF$-stopping time.
We define the associated \emph{default indicator process} $H=(H_t)_{0\leq t\leq \T}$ by $H_t := \ind{\tau\leq t}$, for $t\in[0,\T]$.
We will also make use of the \emph{survival process} $1-H$.
The process $H=\Ind_{\dbra{\tau,\T}}$ is $\bbF$-adapted, bounded, right-continuous and increasing on $[0,\T]$. Hence, by the Doob-Meyer decomposition (see, e.g., \cite[Theorem I.3.15]{JacodShiryaev}), there exists a unique predictable, integrable and increasing process $H^p=(H^p_t)_{0\leq t\leq \T}$ with $H^p_0=0$, called the {\em default compensator} (or dual predictable projection of $H$), such that the process $H-H^p$ is a uniformly integrable martingale on $(\Omega,\bbF,\prob)$. Note also that $H^p_t=H^p_{\tau\wedge t}$, for all $t\in[0,\T]$.

Apart from the minimal assumption of being an $\bbF$-stopping time, we do not introduce any further assumption on $\tau$. In this general setting, the default compensator $H^p$ is not necessarily absolutely continuous (i.e., a {\em default intensity} may fail to exist) and may also contain both singular and jump parts, as shown in the following lemma (proofs are given in Section \ref{sec:proof}).

\begin{lemma}	\label{lem:dec_comp}
The default compensator $H^p$ admits the unique decomposition
\begin{equation}	\label{eq:dec_comp}
H^p_t = \int_0^t h_s ds + \lambda_t + \sum_{0<s\leq t}\Delta H^p_s, 
\qquad\text{ for all } 0\leq t \leq \T,
\end{equation}
where $(h_t)_{0\leq t \leq \T}$ is a non-negative predictable process such that $\int_0^{\T}h_sds<+\infty$ a.s. and $(\lambda_t)_{0\leq t\leq \T}$ is an increasing continuous process with $\lambda_0=0$ such that $d\lambda_s(\omega)\perp ds$, for a.a. $\omega\in\Omega$.
\end{lemma}

We denote by $\{\Delta H^p\neq0\}=\bigcup_{i\in\N}\dbra{U_i}$ the thin set of the jump times of the default compensator $H^p$, where $(U_i)_{i\in\N}$ is a family of predictable times (see \cite[Proposition~I.2.24]{JacodShiryaev}). 
By \cite[\textsection~I.3.21]{JacodShiryaev}, it holds that
\[
\Q(\tau = U_i \leq \T) = \E[\Delta H_{U_i}] = \E[\Delta H^p_{U_i}] >0,
\qquad\text{ for all }i\in\N,
\] 
meaning that the default event has a strictly positive probability of occurrence in correspondence of the predictable dates $(U_i)_{i\in\N}$. The classical \emph{intensity-based approach} can be obtained as a special case by letting $\lambda=\Delta H^p=0$ in decomposition \eqref{eq:dec_comp}.
Typical examples where the continuous singular part $\lambda$ is non-null are provided by last passage times (see, e.g., \cite[Section 4]{EJY}).

\subsection{The term structure of credit risky bonds}	\label{sec:term_str}

A {\em credit risky bond} with maturity $T\in[0,\T]$ is a contingent claim promising to pay one unit of currency at maturity $T$, provided that the defaultable entity does not default before date $T$. 
We denote by $P(t,T)$ the price at date $t$ of a credit risky bond with maturity $T$, for all $0\leq t\leq T\leq \T$.
As a first step, we restrict our attention to the {\em zero-recovery} case, meaning that we assume that the credit risky bond becomes worthless as soon as the default event occurs, i.e., $P(t,T)=0$ if $H_t=1$, for all $0\leq t\leq T\leq \T$ (see Section~\ref{sec:gen_rec} for the analysis of general recovery schemes).

The family of stochastic processes $\{(P(t,T))_{0 \le t \le T}$; $T\in[0,\T]\}$ describes the evolution of the \emph{term structure} $T \mapsto P(t,T)$ over time.
Following the extended HJM-framework suggested in  \cite{GehmlichSchmidt2016MF}, we assume that the term structure of credit risky bonds is of the form
\begin{align}\label{PtT}
  P(t,T) = (1-H_t) \exp\bigg( -\int_t^T f(t,u) du - \int_{(t,T]} g(t,u) \mu_t(du)\bigg), 
  \quad \text{ for all } 0 \le t \le T \le \T,
  \end{align}
corresponding to equation \eqref{HJM:intro2} in the introduction.
Here $\mu(du)=(\mu_t(du))_{0\leq t\leq \T}$ is the measure-valued process defined by $\mu_t(du):=\mu([0,t]\times du)$, for $t\in[0,\T]$, with $\mu(ds,du)$ being the random measure introduced in Section~\ref{sec:proba_setting}.
The processes $f$ and $g$ are assumed to be of the  form\footnote{We want to point out that our results can be extended to the case where $f$ and $g$ are more general semimartingale random fields.
Since our main goal consists in studying defaultable term structures driven by general random measures, we prefer to let $f$ and $g$ be of the simple form \eqref{f1}-\eqref{g1}, in order not to obscure the presentation by too many technical issues.}
\begin{align}	\label{f1}
f(t,T) &= f(0,T) + \int_0^t a(s,T)ds + \int_0^t b(s,T) dW_s,\\
g(t,T) &= g(0,T) + \int_0^t \alpha(s,T)ds + \int_0^t \beta(s,T) dW_s,  
\label{g1}
\end{align}
for all $0\leq t\leq T\leq \T$.
The precise technical assumptions on the random measure $\mu$ as well as on the processes appearing in \eqref{f1}-\eqref{g1} will be given in Section \ref{sec:prelim} below. For the moment, let us briefly comment on the interpretation of the term structure equation \eqref{PtT}.

\begin{remark}[On the role of $g$ and $\mu$]	\label{rem:interpretation}
In comparison with the classical HJM framework applied to credit risk (see, e.g., \cite{DuffieSingleton99,JarrowTurnbull1995,Schoenbucher:CRDerivatives,Schoenbucher98}), the novelty of the term structure equation \eqref{PtT} consists in the presence of the random measure $\mu$ and the associated forward rate $g$. 
The random measure $\mu$ encodes the information received over time about possible ``risky dates'' or ``risky periods'' where, on the basis of the available information, the default event is perceived to be more likely to happen (explicit examples will be provided below).
More specifically, the first argument of $\mu$ represents as usual the running time, while the second argument of $\mu$ identifies the possible risky dates and periods.
Hence, the integral with respect to $\mu_t(du)$ appearing in \eqref{PtT} represents the effect of all the information received up to date $t$ concerning the likelihood of default in the remaining lifetime $(t,T]$ of the bond.  
The assumption that $\mu$ is an optional random measure simply captures the fact that this information about future risky dates is publicly available, but may suddenly arrive in the market (since $\mu$ is not necessarily predictable). 
As will be shown below, absence of arbitrage will imply a precise relationship between the default compensator $H^p$ and the random measure $\mu$.

The forward rate $g$ decodes the impact of this information on the term structure. In some cases it is possible to represent $\int_t^T f(t,u) du +\int_{(t,T]} g(t,u) \mu_t(du)$ by a single term of the form
\begin{align} \label{PtTonlywithf} 
\int_{(t,T]} \tilde f(t,u) \tilde \mu_t(du),
\end{align}
for example when $\mu$ is deterministic, see Section \ref{sec:generalizedMerton}. However, this may not always be feasible or convenient, see Example \ref{ex:doublystochastic}. In this article we decided to cover the general case \eqref{PtT}, while term structure models based on \eqref{PtTonlywithf} clearly allow for a simpler mathematical treatment. \hfill $\diamond$
\end{remark}

The following example illustrates the modelling of bad news which may lead to discontinuities in the term structure. As pointed out in \cite{GehmlichSchmidt2016MF,JiaoLi2015},  the  failure of \euro 1.5 billion of Greece on a scheduled debt repayment to the International Monetary fund as well as Argentina's missed coupon payment on \$29 billion debt are prominent examples of such cases.\footnote{See the announcements in \cite{ISDAArgentina2014} and \cite{ReutersArgentina2014}, as well as \cite{NYTimes2015}.}

\begin{example}[Sovereign credit with surprising bad news]
\label{ex:doublystochastic}
Consider a credit from a country in the best rating class. Under normal circumstances, this could be interpreted as no default risk in the considered time horizon (i.e., $\tau = +\infty$). However, it might be the case that the country is hit by an unexpected event, which could be a catastrophe, a market crash or other unthought risks.
Assume that news about this risk arrive at a random time $S$. The next expected payment of the credit is due at some random time $U>S$ and we denote the probability that the payment will be missed by $p\in[0,1]$. Hence, 
\[
\tau = \begin{cases} U \quad &\text{ with probability }p;\\
+\infty \quad &\text{ with probability }1-p.
\end{cases}
\]
Let the filtration $\bbG=(\cG_t)_{0\leq t\leq \T}$ be generated by the process $(\ind{S \le t} (1+U))_{0\leq t\leq \T}$, properly augmented.
The filtration $\bbF=(\cF_t)_{0\leq t\leq \T}$ is given by the progressive enlargement of $\bbG$ with $\tau$, i.e.,
\[ 
\cF_t = \bigcap_{s>t}(\cG_s \vee \sigma(\tau\wedge s)),
\qquad\text{ for all } 0\leq t\leq \T.
\]  
Then, on $\{t < S\}$, no additional information is available and, hence, $\tau=+\infty$ with probability $(1-p)$. Therefore, for all $A\in\cB([0,\T])\bigcup\{+\infty\}$,
\[ 
\ind{t<S} \Q(\tau \in A| \cG_t) = \ind{t<S} \Big( p \, \Q(U \in A|\cG_t) + (1-p) \delta_\infty(A)\Big)
= \ind{t<S} \Big( p \, \Q(U \in A) + (1-p) \delta_\infty(A)\Big).
\]
Otherwise, on $\{t \ge S\}$, the risky date $U$ is $\cG_t$-measurable, so that
\[
\ind{t \ge S} \Q(\tau \in A| \cG_t)=\ind{t \ge S} \Big( p \, \delta_{U}(A) + (1-p) \delta_\infty(A) \Big),
\]
with $\delta_a$ denoting the Dirac measure in correspondence of point $a$.
This example can be included in our framework by letting
$\mu (ds,du) =  \Ind_{\dbraco{S,+\infty}}(ds) \,\delta_{U}(du)$.
Assume, for simplicity, that the random variable $U$ has a density, so that $\Q(U>T|U>t)$ can be written as an integral with respect to the Lebesgue measure. 
Then, credit risky bond prices following \eqref{PtT} turn out to be of the  HJM form: $P(t,T)=\ind{\tau > t} e^{-\int_t^T f(t,u) du}$ for $t<S$, and
\[ 
P(t,T) = \ind{\tau > t} \exp\bigg(-\int_t^T f(t,u) du - g(t,U) \ind{t < U \le T} \bigg),
\]
for $t\in[S,T]$. 
On the other hand, if the random variable $U$ is discrete, one may consider the generalized Merton model studied in Corollary \ref{cor:Merton}. 
\hfill $\diamond$
\end{example}

\subsection{Technical assumptions and preliminaries}	\label{sec:prelim}

The following assumptions are needed for the analysis of the term structure model, and we begin with  assumptions on the random measure $\mu$. 

\begin{assumption}	\label{ass:rnd_meas}
The random measure $\mu(ds,du)$ is a non-negative optional random measure in the sense of \cite[Definition~II.1.3]{JacodShiryaev} on $[0,\T]\times[0,\T]$ satisfying the following properties:
\begin{enumerate}
\item[(i)] 
$\mu(\omega;ds,du)=\ind{s<u}\mu(\omega;ds,du)$, for all $(s,u)\in[0,\T]\times[0,\T]$ and $\omega\in\Omega$;
\item[(ii)] 
there exists a sequence $(\sigma_n)_{n\in\N}$ of stopping times increasing a.s. to infinity and such that $\E[\mu_{\sigma_n}([0,\T])]<+\infty$ a.s. for every $n\in\N$.
\end{enumerate}
\end{assumption}

According to the interpretation given in Remark~\ref{rem:interpretation}, part (i) of Assumption \ref{ass:rnd_meas} represents the fact that the new information received at date $s$ only concerns the likelihood of default in the future (and not in the past). In view of \eqref{PtT}, this assumption comes without loss of generality.
Part (ii) ensures that the random measure $\mu$ is predictably $\sigma$-finite, in the sense of \cite[Definition II.1.6]{JacodShiryaev}, and that the random variable $\mu_t([0,\T])$ is a.s. finite, for all $t\in[0,\T]$.
Part (ii) of Assumption \ref{ass:rnd_meas} is equivalent to requiring that the increasing process $(\mu_t([0,\T]))_{0\leq t\leq \T}$ is locally integrable (see e.g. \cite[Remark 3.9]{Jacod}).
Apart from Assumption \ref{ass:rnd_meas}, the random measure $\mu$ is allowed to be general. The following lemma presents a first consequence of Assumption \ref{ass:rnd_meas}. We define the process $\bar{\mu}=(\bar{\mu}_t)_{0\leq t\leq \T}$ by
$$ \bar{\mu}_t:=\mu\bigl([0,\T]\times[0,t]\bigr), \qquad \text{ for all } t\in[0,\T]. $$

\begin{lemma}	\label{lem:predictability}
Suppose that Assumption \ref{ass:rnd_meas} holds. 
Then $\bar{\mu}$ is a predictable and increasing process, admitting the unique decomposition
\begin{equation}	\label{eq:dec_pred_proc}
\bar{\mu}_t = \int_0^tm_sds + \nu_t + \sum_{0<s\leq t}\Delta\bar{\mu}_s,
\qquad \text{ for all }0\leq t \leq \T,
\end{equation}
where  $(m_t)_{0\leq t \leq \T}$ is a  non-negative predictable process such that $\int_0^{\T}m_sds<+\infty$ a.s. and $(\nu_t)_{0\leq t\leq \T}$ is an increasing continuous process with $\nu_0=0$ such that $d\nu_s(\omega)\perp ds$, for almost all $\omega\in\Omega$.
\end{lemma}

The random variable $\bar{\mu}_t$ measures the existence of risky dates in the period $[0,t]$, on the basis of all available information over $[0,\T]$, compare Remark~\ref{rem:interpretation}.
In a similar way, the quantity $\Delta\bar{\mu}_t=\mu([0,\T]\times\{t\})$ encodes whether time $t$ is perceived as a risky date, on the basis of all available information.
As we shall see in Theorem \ref{thm:dc}, the absence of arbitrage implies a precise relationship between the terms appearing in the  decompositions \eqref{eq:dec_comp} and \eqref{eq:dec_pred_proc}.

The following mild technical assumptions ensure that all (stochastic) integrals are well-defined and that we can apply suitable versions of the (stochastic) Fubini theorem. In the following, we denote by $\cO$ ($\cP$, resp.) the optional (predictable, resp.) $\sigma$-field on $(\Omega,\cA,\bbF)$.

\pagebreak

\begin{assumption}	\label{ass:processes}
The following conditions hold a.s.:
\begin{enumerate}[(i)]
\item
The {\em initial forward curves} $T\mapsto f(\omega;0,T)$ and $T\mapsto g(\omega;0,T)$ are $\cF_0\otimes\cB([0,\T])$-measurable, real-valued, continuous and integrable on $[0,\T]$:
\[
\int_0^{\T}|f(0,u)|du<+\infty
\quad\text{and}\quad
\int_0^{\T}|g(0,u)|\mu_t(du) < +\infty,
\quad\text{ for all }t\in[0,\T];
\]
\item
the \emph{drift processes} $a(\omega;s,u)$ and $\alpha(\omega;s,u)$ are $\mathcal{O}\otimes\mathcal{B}([0,\T])$-measurable and real-valued,  $a(\omega;s,u)=0$ and $\alpha(\omega;s,u)=0$ for all $0\leq u<s\leq \T$, the maps $u\mapsto a(\omega;s,u)$ and $u\mapsto\alpha(\omega;s,u)$ are differentiable and satisfy 
\begin{gather*}
\int_0^{\T}\int_0^{\T}\bigl(|a(s,u)|+|\partial_u a(s,u)|\bigr)ds\,du < +\infty,	\\
\int_0^{\T}\int_0^{\T}|\partial_u\alpha(s,u)|ds\,du < +\infty
\quad\text{ and }\quad
\int_0^{\T}\int_0^{\T}|\alpha(s,u)|ds\mu_t(du) < +\infty,
\quad\text{ for all }t\in[0,\T];
\end{gather*}
\item
the \emph{volatility processes} $b(\omega;s,u)$ and $\beta(\omega;s,u)$ are $\mathcal{O}\otimes\mathcal{B}([0,\T])$-measurable and $\R^n$-valued, $b(\omega;s,u)=0$ and $\beta(\omega;s,u)=0$ for all $0\leq u<s\leq \T$, the maps $u\mapsto b(\omega;s,u)$ and $u\mapsto\beta(\omega;s,u)$ are differentiable and satisfy 
\begin{gather*}
\int_0^{\T}\int_0^{\T}\bigl(\|b(s,u)\|+\|\partial_ub(s,u)\|\bigr)^2ds\,du < +\infty\\
\int_0^{\T}\int_0^{\T}\|\partial_u\beta(s,u)\|^2ds\,du < +\infty
\quad\text{ and }\quad
\int_0^{\T}\int_0^{\T}\bigl\|\beta(s,u)\bigr\|^2\mu_s(du)ds < +\infty
\end{gather*}
\end{enumerate}
\end{assumption}

It is easy to check that Assumption~\ref{ass:processes} implies that both the ordinary and the stochastic integrals appearing in \eqref{f1}-\eqref{g1} are well-defined.
Arguing similarly as in the proof of \cite[Proposition 6.1]{Filipovic2009}, the processes $(f(t,t))_{0\leq t\leq \T}$ and $(g(t,t))_{0\leq t\leq \T}$ are continuous, hence predictable and locally bounded.
By an analogous argument, it can be shown that, for any fixed $t\in[0,\T]$ and for a.a. $\omega\in\Omega$, the maps $u\mapsto f(\omega;t,u)$ and $u\mapsto g(\omega;t,u)$ are continuous.
Together with part (ii) of Assumption~\ref{ass:rnd_meas}, this implies that both integrals appearing in the term structure equation~\eqref{PtT} are a.s. finite.
Finally, $\int_0^{\cdot}g(s,s)d\bar{\mu}_s$ is well-defined as a predictable process of finite variation.

\section{HJM-type conditions for general defaultable term structures}	\label{sec:main_results}

This section contains our main results and provides a characterization of the absence of arbitrage in the context of general term structure models as introduced in Section~\ref{sec:tau}. After making precise in Section~\ref{sec:arbitrage} the description of the financial market and the notion of arbitrage we consider, we discuss in Section~\ref{sec:easy_case} a simple formulation of our main result, in order to illustrate in a transparent and easy setting the meaning of the HJM-type conditions. 
Section~\ref{sec:main_result} contains the statement of our main theorem, while Section~\ref{sec:examples} deals with several special cases of interest and Section~\ref{sec:gen_rec} presents the extension to general recovery schemes. Finally, we discuss related works in Section~\ref{sec:rel-literature}. The proofs of all the results are given in Section~\ref{sec:proof}.

\subsection{The large financial market of credit risky bonds}	\label{sec:arbitrage}

The considered financial market is assumed to contain a \emph{num\'eraire}, whose price process is   strictly positive, c\`adl\`ag and adapted, and denoted by $X^0=(X^0_t)_{0\leq t\leq \T}$. Without loss of generality, we assume that $X^0_0=1$. Moreover, we make the classical assumption that $X^0$ is absolutely continuous, i.e., there exists a predictable integrable short rate process $r=(r_t)_{0\leq t\leq \T}$ such that $X^0_t=\exp(\int_0^t r_s ds)$, for all $t\in[0,\T]$. For practical applications, one would typically use the overnight index swap (OIS) rate for constructing the num\'eraire. 

The credit risky bond market consists of the uncountable family $\{(P(t,T))_{0\leq t\leq T};T\in[0,\T]\}$, representing the price processes of all basic traded assets. In particular, this financial market is infinite-dimensional and, therefore, can be treated as a {\em large financial market}, in the spirit of \cite{CuchieroKleinTeichmann,KleinSchmidtTeichmann2015}.
This corresponds to considering sequences of trading strategies, each strategy only consisting of portfolios of finitely many but arbitrary credit risky bonds, and taking the limits of those. If the limit is taken in \'Emery's semimartingale topology, this leads to the notion of {\em no asymptotic free lunch with vanishing risk (NAFLVR)} recently introduced in a general setting in \cite{CuchieroKleinTeichmann}. In particular, NAFLVR holds in our context  if the probability measure $\Q$ is a {\em local martingale measure} for the family $\{(P(t,T))_{0\leq t\leq T};T\in[0,\T]\}$ with respect to the num\'eraire $X^0$, i.e., if the process $P(\cdot,T)/X^0$ is a $\Q$-local martingale, for every $T\in[0,\T]$.
In the following, we will derive necessary and sufficient conditions for this property to hold, thus ensuring that the credit risky bond market is arbitrage-free in the sense of NAFLVR.

\begin{remark}
{\em (i)}
The present setting can be extended to consider the case where the num\'eraire $X^0$ is a general strictly positive semimartingale (not necessarily absolutely continuous), along the lines of \cite{KleinSchmidtTeichmann2015}. In this case, one can obtain generalized versions of Theorems~\ref{thm:dc} and \ref{thm:dc2}, at the expense of a more complex formulation.

{\em (ii)}
Under the additional hypothesis of locally bounded bond prices, \cite[Theorem 5]{KleinSchmidtTeichmann2015} shows that the existence of an equivalent local martingale measure is equivalent to the {\em no asymptotic free lunch (NAFL)} condition. 
In the present setting, NAFL is equivalent to NAFLVR. This implies that, under this additional hypothesis, our results allow to deduce necessary and sufficient conditions for NAFLVR to hold in the large financial market of credit risky bonds.
However, in the general setting introduced in Section~\ref{sec:tau}, the local boundedness property does not necessarily hold.
\end{remark}

\subsection{A first result}	\label{sec:easy_case}

Before stating our main result (Theorem~\ref{thm:dc} below), let us first present a special case which sheds some light on the no-arbitrage restrictions for a defaultable term structure model of the form \eqref{PtT}-\eqref{g1}. 
Consider an {\em integer-valued} random measure $\mu(ds,du)$ in the sense of \cite[Definition II.1.13]{JacodShiryaev}, satisfying Assumption \ref{ass:rnd_meas}. 
Then, by \cite[Proposition II.1.14]{JacodShiryaev}, there exist a sequence of stopping times $(\sigma_n)_{n\in\N}$ and a $[0,\T]$-valued optional process $\gamma=(\gamma_t)_{0\leq t\leq \T}$ such that, defining the $\cF_{\sigma_n}$-measurable random variable $\tau_n:=\gamma_{\sigma_n}$, for all $n\in\N$, it holds that
\begin{align}\label{eq:repmuintegervalued}
\mu(ds,du)
= \sum_{n=1}^{+\infty}\delta_{(\sigma_n,\tau_n)}(ds,du).
\end{align}
Note that, since $\mu(ds,du)=\ind{s<u}\mu(ds,du)$ by Assumption \ref{ass:rnd_meas}, we have that $\tau_n>\sigma_n$, for all $n\in\N$.
This setting allows for an intuitive economic interpretation:
\begin{itemize}
\item each stopping time $\sigma_n$ represents an {\em announcement date}. At an announcement date  new information concerning the future likelihood of default is released in the market. Since $\sigma_n$ is a general stopping time, announcements may come as a surprise;
\item each stopping time $\tau_n$ represents a {\em risky date} in the future which is revealed at the announcement date $\sigma_n$. A risky date is a date where default is perceived to be possible with positive probability. This date itself does not come as a surprise (as $\tau_n$ is $\cF_{\sigma_n}$-measurable and $\tau_n>\sigma_n$), while whether default occurs at $\tau_n$ or not remains unpredictable in general. 
\end{itemize}
Such risky dates naturally lead to discontinuities in the term structure, thus violating \eqref{HJM:intro}.
The representation \eqref{eq:repmuintegervalued} allows to simplify the defaultable term structure equation \eqref{PtT} to
\[
P(t,T) = (1-H_t)\exp\left(-\int_t^Tf(t,u)du 
- \sum_{\sigma_n\leq t}\ind{\tau_n\in(t,T]}g(t,\tau_n)\right),
\quad 0\leq t\leq T\leq \T.
\]
Note that the price $P(t,T)$ of a defaultable bond depends on all the announcements  received up to date $t$ concerning future risky dates  in the remaining lifetime $(t,T]$ of the bond.
In this setting, we can formulate the following proposition, which represents a special case of Corollary \ref{cor:int_valued} below (see Section \ref{subsec:ivmu}).
We summarize as follows the assumptions introduced in the present subsection.

\begin{assumption}\label{ass:simplecase}
Suppose  that
\begin{enumerate}[(i)]
\item
$\mu(ds,du)$ is an integer-valued random measure in the representation \eqref{eq:repmuintegervalued} and with a compensator of the form $\mu^p(ds,du)=\xi_s(du)ds$, where $\xi_s(du)$ is a positive finite measure on $([0,\T],\cB([0,\T]))$, for all $s\in[0,\T]$;
\item
$\prob(\tau=\sigma_n)=0$, for all $n\in\N$;
\item 
the continuous singular process $\lambda$ appearing in \eqref{eq:dec_comp} vanishes.
\end{enumerate}
\end{assumption}

In particular, part (ii) of Assumption \ref{ass:simplecase} requires that no new information arrives simultaneously with the default event. 
We set, for all $0\leq t \leq T\leq \T$,
\begin{equation}	\label{eq:notation_integrals}	\begin{aligned}
\bar{a}(t,T) &= \int_t^T a(t,u) du, &\quad&
\bar{b}(t,T) = \int_t^T b(t,u) du, \\
\bar{\alpha}(t,T)&=\int_{(t,T]} \alpha(t,u)\mu_t(du), &\quad&
\bar{\beta}(t,T)=\int_{(t,T]} \beta(t,u)\mu_t(du).
\end{aligned}	\end{equation}
Note that, as long as Assumptions \ref{ass:rnd_meas} and \ref{ass:processes} are satisfied, all the above integrals are well-defined. 

\begin{proposition}	\label{prop:easy_case}
Suppose that Assumptions \ref{ass:rnd_meas}, \ref{ass:processes} and \ref{ass:simplecase} hold.
Then, the probability measure $\Q$ is a local martingale measure for $\{(P(t,T))_{0\leq t\leq T};T\in[0,\T]\}$ with respect to $X^0$ if and only if the following conditions hold a.s.:
\begin{enumerate}[(i)]
\item 
$f(t,t) = r_t+h_t$, for Lebesgue-a.e. $t\in[0,\T]$;
\item
$\{\Delta H^p\neq0\}\subseteq\bigcup_{n\in\N}\dbra{\tau_n}$ and $\Delta H^p_{\tau_n}=1-e^{-g(\tau_n,\tau_n)}$, for all $n\in\N$;
\item 
for all $T\in[0,\T]$ and Lebesgue-a.e. $t\in[0,T]$, it holds that
\[
-\bar{a}(t,T)
- \bar{\alpha}(t,T) 
+ \frac{1}{2}\|\bar{b}(t,T)+\bar{\beta}(t,T)\|^2
+\int_t^T\bigl(e^{-g(t,u)}-1\bigr)\xi_t(du)
 = 0.
\]
\end{enumerate}
\end{proposition}

The necessary and sufficient conditions stated in Proposition \ref{prop:easy_case} admit a clear interpretation:
\begin{itemize}
\item condition {\em (i)} requires the instantaneous yield $f(t,t)$ on the defaultable bond to be equal to the risk-free rate $r_t$ plus a default risk compensation term given by $h_t$. This corresponds to a classical no-arbitrage restriction in intensity-based models (see, e.g., \cite[Theorem 2]{Schoenbucher98});
\item condition {\em (ii)} requires that all the predictable times at which the default event can happen with strictly positive probability are announced as risky dates. Moreover, the second part of condition {\em (ii)} means that, at every risky date $\tau_n$, the defaultable bond price exhibits a jump  satisfying $\E[\Delta P(\tau_n,T)|\cF_{\tau_n-}]=0$ a.s.;
\item condition {\em (iii)} corresponds to the classical HJM drift restriction.
The additional term $\int_t^T(e^{-g(t,u)}-1)\xi_t(du)$ represents a compensation for the movements in the term structure due to the arrival of news concerning possible future risky dates.
\end{itemize}

\subsection{The main result}	\label{sec:main_result}

Let us now proceed to the statement of our main result, exploiting the full generality of the setting.
For each $T\in[0,\T]$, we introduce the process $Y^{(T)}=(Y^{(T)}_t)_{0\leq t\leq T}$ defined by  
\begin{equation}	\label{eq:YtT}
Y^{(T)}_t := \int_0^t\int_0^Tg(s,u)\mu(ds,du), 
\qquad\text{ for all } 0\leq t\leq T.
\end{equation}
It will come as a consequence of Lemma~\ref{lem:compute_G} that  $Y^{(T)}$ is a.s. finite and well-defined as a finite variation process.
We denote by $\mu^{(Y^{(T)},H)}$ the jump measure associated to the two-dimensional semimartingale $(Y^{(T)},H)$, in the sense of \cite[Proposition II.1.16]{JacodShiryaev}, with compensator $\mu^{p,(Y^{(T)},H)}$.   
By \cite[Theorem II.1.8]{JacodShiryaev}, there exist an increasing integrable predictable process $A^{(T)}=(A^{(T)}_t)_{0\leq t\leq \T}$ and a kernel $K^{(T)}(\omega,t;dy,dz)$ from $(\Omega\times[0,\T],\cP)$ into $(\R\times\{0,1\},\cB(\R\times\{0,1\}))$ such that
\begin{equation}	\label{eq:kernel1}
\mu^{p,(Y^{(T)},H)}(\omega;dt,dy,dz) = K^{(T)}(\omega,t;dy,dz) dA^{(T)}_t(\omega).
\end{equation}
Due to part (i) of Assumption \ref{ass:rnd_meas}, it holds that $\Delta Y^{(T)}_T=Y^{(T)}_T-Y^{(T)}_{T-}=\int_0^Tg(T,u)\mu(\{T\}\times du)=0$. Hence, we can assume without loss of generality that $\ind{y\neq0}K^{(T)}(\omega,T;dy,dz)=0$, for all $(\omega,T)\in\Omega\times[0,\T]$.
Let $\mu^p$ be the compensator of $\mu$, which exists by part (ii) of Assumption~\ref{ass:rnd_meas} together with \cite[Theorem II.1.8]{JacodShiryaev}.
As shown in Lemma \ref{lem:rel_comp_meas} below, the compensating measure $\mu^{p,(Y^{(T)},H)}$ is linked to $\mu^p$ via the following relation, for all $0\leq t \leq T\leq \T$:
\begin{equation}	\label{eq:link_comp_meas}
 \Delta A_t^{(T)}\int_{\mathbb{R}}y\,K^{(T)}(t;dy\times\{0,1\})
= \int_{\mathbb{R}}y\,\mu^{p,(Y^{(T)},H)}(\{t\}\times dy\times\{0,1\})
= \int_{(t,T]}g(t,u)\,\mu^p(\{t\}\times du).
\end{equation}

Finally, we define the function $\Psi:\Omega\times[0,\T]\times\mathbb{R}\times[0,1]\rightarrow\mathbb{R}$ as
\begin{equation}	\label{eq:def_W}	\begin{aligned}
\Psi(\omega;t,y,z) &:= e^{g(\omega;t,t)\Delta\bar{\mu}_t(\omega)}
\left(e^{-y}-1\right)( 1 - z).
\end{aligned}	\end{equation}
Note that, since the two processes $(g(t,t))_{0\leq t\leq \T}$ and $(\bar{\mu}_t)_{0\leq t\leq \T}$ are predictable, the function $\Psi$ is $\mathcal{P}\otimes\mathcal{B}(\mathbb{R}\times[0,1])$-measurable.
Following the notation of \cite{JacodShiryaev}, we denote by ``$*$'' integration with respect to a random measure. 
Moreover, for an arbitrary process $V=(V_t)_{0\leq t\leq \T}$ of finite variation, we denote by $V^c$ its continuous part, which can be further decomposed as $V^c=\int_0^{\cdot} V^{\text{ac}}_sds + V^{\text{sg}}$, similarly as in Lemma \ref{lem:dec_comp}.
We will also make use of the notation introduced in \eqref{eq:notation_integrals}.

We are now in a position to state the following theorem, which  gives necessary and sufficient conditions rendering the reference probability measure $\Q$ a local martingale measure for the family $\{(P(t,T))_{0\leq t\leq T};T\in[0,\T]\}$ with respect to the num\'eraire $X^0$. As mentioned above, this represents a cornerstone for ensuring absence of arbitrage in the sense of NAFLVR.
The proof of the theorem will be given in Section \ref{sec:proof_thm}.

\begin{theorem}\label{thm:dc}
Suppose that Assumptions \ref{ass:rnd_meas} and \ref{ass:processes} hold. Let $\Psi$  be defined as in \eqref{eq:def_W} and $g^{(T)}(\omega;s,u):=\ind{u\leq T}g(\omega;s,u)$, for each $T\in[0,\T]$. 
Then, the probability measure $\Q$ is a local martingale measure for $\{(P(t,T))_{0\leq t\leq T};T\in[0,\T]\}$ with respect to $X^0$ if and only if the following conditions hold a.s.:
\begin{enumerate}[(i)]
\item 
$
f(t,t) + g(t,t)m_t = r_t+h_t
$, 
for Lebesgue-a.e. $t\in[0,\T]$;
\item 
$\Delta H^p_t=1-e^{-g(t,t)\Delta\bar{\mu}_t}$, for all $t\in[0,\T]$;
\item 
$ \Delta (\Psi  * \mu^{p,(Y^{(T)},H)})_t = 0$ , for all $0\leq t\leq T\leq \T$.
\item
for all $T\in[0,\T]$ and for Lebesgue-a.e. $t\in[0,T]$, it holds that
\begin{align*}
&-\bar{a}(t,T)
- \bar{\alpha}(t,T)
+ \frac{1}{2}\|\bar{b}(t,T)+\bar{\beta}(t,T)\|^2	
-(g^{(T)} * \mu)^{{\rm ac}}_t	
+(\Psi * \mu^{p,(Y^{(T)},H)})^{{\rm ac}}_t 
= 0;
\end{align*}
\item 
for all $0\leq t\leq T\leq \T$, it holds that
\[
\int_0^tg(s,s)d\nu_s - (g^{(T)} * \mu)_t^{\text{\rm sg}} +  (\Psi * \mu^{p,(Y^{(T)},H)})_t^{\text{\rm sg}} = \lambda_t.
\]
\end{enumerate} 
\end{theorem}

Due to the generality of the setting, the conditions given in Theorem \ref{thm:dc} are rather complex. However, as we are now going to explain, each of them has a precise financial interpretation, similarly to the case of Proposition \ref{prop:easy_case}. These conditions will be further discussed in Section \ref{sec:examples} in the context of several examples and special cases of practical interest. 

Condition \emph{(i)} in Theorem \ref{thm:dc} requires the instantaneous yield $f(t,t)+g(t,t)m_t$ accumulated by the credit risky bond to be equal to the risk-free rate of interest $r_t$ plus a default risk compensation term given by $h_t$, which corresponds to the density of the absolutely continuous part of the default compensator $H^p$ (see Lemma \ref{lem:dec_comp}). This condition is therefore analogous to condition {\em (i)} in Proposition \ref{prop:easy_case}.

Condition \emph{(ii)} is a precise matching condition between the jumps of the default compensator $H^p$ and the jumps of the process $\bar{\mu}$ introduced in Lemma \ref{lem:predictability}. 
In particular, letting the predictable times $(U_i)_{i\in\N}$ represent the jump times of $H^p$, condition {\em (ii)} implies that
\begin{align}
\bigl\{(\omega,t)\in\Omega\times[0,\T] : g(\omega;t,t)\Delta\bar{\mu}_t(\omega)\neq0\bigr\}
= \bigl\{(\omega,t)\in\Omega\times[0,\T] : \Delta H^p_t(\omega)\neq0\bigr\}	
= \bigcup_{i\in\mathbb{N}}\dbra{U_i}.
\label{eq:jump_times}
\end{align}
Since the predictable times $(U_i)_{i\in\N}$ correspond to possible default dates (i.e., $\Q(\tau=U_i\leq\T)>0$, for all $i\in\N$) and the jumps of $\bar{\mu}$ correspond to ``risky dates'', relation \eqref{eq:jump_times} means that ``false alarms'' (i.e., the possibility that a date for which there is no possibility of default is announced as a risky date) cannot happen.
Moreover, observe that condition {\em (ii)} implies that, in order to exclude arbitrage, the credit risky term structure must exhibit discontinuities in maturity at the jump times $(U_i)_{i\in\N}$ of the default compensator. In other words, credit risky bond prices must be discontinuous in correspondence of the risky dates (recall that $\Delta\bar{\mu}_t=\mu([0,\T]\times\{t\})$, for all $t\in[0,\T]$).
This condition is analogous to condition {\em (ii)} in Proposition \ref{prop:easy_case}.

Condition {\em (iii)} requires that the overall effect of new information about  future risky dates arriving at predictable times and not coinciding with the default event vanishes. This can be seen by rewriting condition {\em (iii)} in the equivalent form
\begin{align*}
0 &= \Delta(\Psi  * \mu^{p,(Y^{(T)},H)})_t	
= e^{g(t,t)\Delta\bar{\mu}_t}\Delta A^{(T)}_t\int_{\R\times\{0,1\}}(e^{-y}-1)(1-z)K^{(T)}(t;dy,dz)	\\
&= \E\left[e^{g(t,t)\Delta\bar{\mu}_t}\bigl(e^{-\Delta Y^{(T)}_t}-1\bigr)\bigl(1-\Delta H_t\bigr)\bigr|\cF_{t-}\right]	\\
& \!\!\!\!\Longleftrightarrow 
\quad\E\left[\left(e^{-\int_t^Tg(t,u)\mu(\{t\}\times du)}-1\right)\bigl(1-\Delta H_t\bigr)\bigr|\cF_{t-}\right] = 0,
\end{align*}
where we have used  \cite[\textsection~II.1.11]{JacodShiryaev} and the predictability of the processes $(g(t,t))_{0\leq t\leq \T}$ and $\bar{\mu}$.
Note that this condition is always satisfied if the process $Y^{(T)}$ is {\em quasi-left-continuous}, for every $T\in[0,\T]$ (see \cite[Corollary II.1.19]{JacodShiryaev}).

Condition \emph{(iv)} represents the extension to a general defaultable setting of the classic HJM drift condition. This condition is analogous to condition {\em (iii)} in Proposition \ref{prop:easy_case}. 
Let us decompose 
\begin{align*}  \begin{split}
\Psi(\omega;t,y,z) &= e^{g(\omega;t,t)\Delta\bar{\mu}_t(\omega)}
\left(e^{-y}-1\right) - e^{g(\omega;t,t)\Delta\bar{\mu}_t(\omega)}z
\left(e^{-y}-1\right) \\
&=: \Psi^{(1)} (\omega;t,y) - \Psi^{(2)}(\omega;t,y,z), \end{split}
\end{align*}
so that
\begin{equation}\label{decomp:W}
\Psi * \mu^{p,(Y^{(T)},H)}
= \Psi^{(1)} * \mu^{p,Y^{(T)}} - \Psi^{(2)} * \mu^{p,(Y^{(T)},H)},
\end{equation}
where $\mu^{p,Y^{(T)}}$ is the compensator of the jump measure $\mu^{Y^{(T)}}$ of $Y^{(T)}$, for $T\in[0,\T]$.
The terms $\Psi^{(1)} * \mu^{p,Y^{(T)}}$ and $g^{(T)} * \mu$ represent a compensation for the information received at time $t$ concerning the likelihood of default in the future time period $(t,T]$, while the term $\Psi^{(2)} * \mu^{p,(Y^{(T)},H)}$ accounts for the possibility of news arriving simultaneously to the default event.

Finally, condition {\em (v)} relates the continuous singular part $\lambda$ of the default compensator $H^p$ to the continuous singular processes appearing in the semimartingale decomposition of the term $\int_{(t,T]}g(t,u)\mu_t(du)$ in \eqref{PtT}.
Note also that, making use of \eqref{eq:kernel1}, the term $ (\Psi * \mu^{p,(Y^{(T)},H)})_t^{\text{\rm sg}}$ appearing in condition {\em (v)} can be equivalently rewritten as
\[
(\Psi* \mu^{p,(Y^{(T)},H)})_t^{\text{\rm sg}}
= \int_0^t\int_{\R\times\{0,1\}}(e^{-y}-1)(1-z)K^{(T)}(s;dy,dz)dA^{(T),{\rm sg}}_s.
\]

\begin{remark}[On the impossibility of predictable default]	\label{rem:pred_default}
Condition {\em (ii)} of Theorem \ref{thm:dc} implies that $\Delta H^p_t<1$ a.s., for all $t\in[0,\T]$, since the term $g(t,t)\Delta\bar{\mu}_t$ is a.s. finite. In particular, this implies that the default time $\tau$ cannot be a {\em predictable time} (in the sense of \cite[Definition I.2.7]{JacodShiryaev}).
Intuitively, it is clear why a predictable default time is incompatible with an arbitrage-free term structure of the form \eqref{PtT}, provided that $\int_t^Tf(t,u)du+\int_{(t,T]}g(t,u)\mu_t(du)<+\infty$ a.s., for all $0\leq t\leq T\leq \T$. 
Indeed, if $\tau$ was a predictable time with $\Q(\tau\leq T)>0$, for some $T\in[0,\T]$, then the elementary strategy $-\Ind_{\dbra{\tau}}\ind{\tau\leq T}$ would realize an arbitrage opportunity, since $P(\tau,T)=0$ and $P(\tau-,T)>0$ hold on $\{\tau\leq T\}$. 
This is related to the fact that absence of arbitrage necessarily excludes jumps of predictable size occurring at predictable times (see \cite{FPP}).
Note, however, that this argument does only exclude the case where the default time $\tau$ is a predictable time, but does not exclude the case where $\tau$ can occur with strictly positive (but not unit) probability at predictable times (i.e., $\tau$ can be an {\em accessible time}, see \cite[Definition 3.34]{HWY}).
\end{remark}

\subsection{Special cases}\label{sec:examples}

In this section, we present several special cases of Theorem  \ref{thm:dc} of practical interest. Further examples related to the existing literature are given in Section \ref{sec:rel-literature}.
We start with the following simple lemma, which shows that conditions {\em (iii)}-{\em (iv)}-{\em (v)} of Theorem \ref{thm:dc} can be simplified under a rather mild additional assumption on the news arrival process, encoded by the random measure $\mu$.

\begin{lemma}	\label{lem:no_news_default}
Suppose that Assumptions \ref{ass:rnd_meas} and \ref{ass:processes} hold and assume furthermore that
\begin{equation}	\label{eq:no_news_default}
\mu\bigl(\{t\}\times[0,\T]\bigr)\Delta H_t=0 \text{ a.s. }
\quad\text{ for all }t\in[0,\T].
\end{equation}
Then, the term $\Psi * \mu^{p,(Y^{(T)},H)}$ appearing in conditions (iii)-(iv)-(v) of Theorem \ref{thm:dc} coincides with the term $\Psi^{(1)} * \mu^{p,Y^{(T)}}$ introduced in \eqref{decomp:W},  for every $T\in[0,\T]$.
\end{lemma}

In particular, recalling decomposition \eqref{decomp:W}, this lemma shows that the term $\Psi^{(2)} * \mu^{p,(Y^{(T)},H)}$ only plays the role of a compensation for the risk of news arriving simultaneously to the default event. 
Condition \eqref{eq:no_news_default} can equivalently be phrased as a ``no default by news'' condition. In view of practical applications, this certainly represents a plausible assumption.

\subsubsection{Integer-valued random measures} \label{subsec:ivmu} 

We now consider  the case where condition \eqref{eq:no_news_default} holds and the random measure $\mu(ds,du)$ is {\em integer-valued}. As already explained in Section \ref{sec:easy_case}, this additional assumption corresponds to the situation where, at each date $t$, the new information arriving at that date only concerns a single time point (a possible ``risky date'') in the future time period $(t,T]$, see equation \eqref{eq:repmuintegervalued}. In view of practical applications, this case is still sufficiently general and, as shown in Corollary \ref{cor:int_valued} below, allows for a substantial simplification of Theorem \ref{thm:dc}. 
In particular, Proposition \ref{prop:easy_case} will follow as a direct consequence of Corollary \ref{cor:int_valued}.

Recall from Section \ref{sec:easy_case}, that for  an integer-valued random measure  $\mu(ds,du)$ on $[0,\T]\times[0,\T]$ satisfying Assumption \ref{ass:rnd_meas}, there exist a thin random set $D=\bigcup_{n\in\mathbb{N}}\dbra{\sigma_n}$ and associated random variables $(\tau_n)_{n\in\N}$, where $\tau_n$ is $\cF_{\sigma_n}$-measurable and $\tau_n > \sigma_n$, such that representation \eqref{eq:repmuintegervalued} holds.
Clearly, condition \eqref{eq:no_news_default} holds if and only if $\prob(\tau=\sigma_n)=0$ for all $n\in\N$.
By \cite[Theorem II.1.8]{JacodShiryaev}, the compensator $\mu^p(ds,du)$  of $\mu(ds,du)$ admits a decomposition of the form 
\[
\mu^p(\omega;ds,du)= F(\omega,s;du)dJ_s(\omega),
\]
where $J=(J_t)_{0\leq t\leq \T}$ is an increasing integrable predictable process and $F(\omega,s;du)$ is a kernel from $(\Omega\times[0,\T],\cP)$ into $([0,\T],\cB([0,\T]))$.

\begin{corollary}	\label{cor:int_valued}
Suppose that Assumptions \ref{ass:rnd_meas} and \ref{ass:processes} hold and suppose furthermore that condition \eqref{eq:no_news_default} holds and that the random measure $\mu(ds,du)$ is integer-valued. 
Then, the probability measure $\Q$ is a local martingale measure for $\{(P(t,T))_{0\leq t\leq T};T\in[0,\T]\}$ with respect to $X^0$ if and only if the following conditions hold a.s.:
\begin{enumerate}[(i)]
\item 
$f(t,t) = r_t+h_t$, for Lebesgue-a.e. $t\in[0,\T]$;
\item
$\{\Delta H^p\neq0\}\subseteq\bigcup_{n\in\N}\dbra{\tau_n}$ and $\Delta H^p_{\tau_n}=1-e^{-g(\tau_n,\tau_n)}$, for all $n\in\N$;
\item 
$\Delta J_t\int_{(t,T]}(e^{-g(t,u)}-1)F(t;du)=0$, for all $0\leq t\leq T\leq \T$;
\item 
for all $T\in[0,\T]$ and Lebesgue-a.e. $t\in[0,T]$, it holds that
\[
-\bar{a}(t,T)
- \bar{\alpha}(t,T) 
+ \frac{1}{2}\|\bar{b}(t,T)+\bar{\beta}(t,T)\|^2
+J^{{\rm ac}}_t\int_t^T(e^{-g(t,u)}-1)F(t;du)
 = 0;
\]
\item
$\int_0^t\int_{(s,T]}(e^{-g(s,u)}-1)F(s;du)dJ^{{\rm sg}}_s = \lambda_t$, for all $0\leq t\leq T\leq \T$.
\end{enumerate}
\end{corollary}

In particular, the additional assumptions that condition \eqref{eq:no_news_default} holds and that $\mu$ is integer-valued imply that the default compensator $H^p$ can have a singular part $\lambda$ if and only if the compensating measure $\mu^p$ admits a singular part (condition {\em (v)} of the above corollary). 
Furthermore, condition {\em (i)} simply requires the short end of the riskless forward rate $f(t,t)$ to be equal to the risk-free rate $r_t$ plus the instantaneous compensation $h_t$ for the risk of default.
Observe also that, in the above corollary, the formulation of the necessary and sufficient conditions does not require the introduction of the auxiliary jump measure $\mu^{p,(Y^{(T)},H)}$.
Moreover, if the compensator $\mu^p$ has the form $\mu^p(ds,du)=\xi_s(du)ds$, for some positive finite measure $\xi_s(du)$ (this is for instance the case when $\mu$ is a {\em homogeneous Poisson random measure}, see \cite[Definition II.1.20]{JacodShiryaev}), then conditions {\em (iii)} and {\em (v)} of Corollary \ref{cor:int_valued} are automatically satisfied if $\lambda=0$. In view of these observations, Proposition \ref{prop:easy_case} follows as a special case of the above corollary.
 
Theorem \ref{thm:dc} allows to recover the two special cases originally considered in \cite{GehmlichSchmidt2016MF}.
The first corollary below considers a tractable setting where there is a finite number $(\tau_n)_{n=1,\dots,N}$ of risky dates (defined in Section \ref{sec:easy_case}) at which default can occur with strictly positive probability and each of which is publicly announced at some previous announcement time $\sigma_n$, $1 \le n \le N$. 
This result can also be seen as a generalization of Example \ref{ex:doublystochastic}.

\begin{corollary}	\label{cor:GS_case1}
Suppose that Assumption \ref{ass:processes} holds and assume furthermore that
\begin{enumerate}[(a)]
\item
the default compensator $H^p$ satisfies $\lambda=0$ and $\{\Delta H^p\neq 0\}=\bigcup_{i=1}^N\dbra{\tau_i}$, for some $N\in\N$, and $\Delta H^p_{\tau_i}$ is $\cF_{\sigma_i}$-measurable, where $(\sigma_i)_{i=1,\ldots,N}$ is a sequence of strictly increasing stopping times such that $\sigma_i<\tau_i$ a.s., for all $i=1,\ldots,N$;
\item 
$\mu(ds,du)=\sum_{i=1}^N\delta_{\{\sigma_i,\tau_i\}}(ds,du)$;
\item
the compensator $\mu^p$ has the form $\mu^p(ds,du)=\xi_s(du)ds$;
\item
$\Q(\tau=\sigma_i)=0$, for all $i=1,\ldots,N$.
\end{enumerate} 
Then, the probability measure $\Q$ is a local martingale measure for $\{(P(t,T))_{0\leq t\leq T};T\in[0,\T]\}$ with respect to $X^0$ if and only if the following conditions hold a.s.:
\begin{enumerate}[(i)]
\item
$f(t,t)=r_t+h_t$, for Lebesgue-a.e. $t\in[0,\T]$;
\item
$\Delta H^p_{\tau_i}=1-e^{-g(\tau_i,\tau_i)}$, for all $i=1,\ldots,N$;
\item
for all $T\in[0,\T]$ and Lebesgue-a.e. $t\in[0,T]$, it holds that
\[
-\bar{a}(t,T)
- \bar{\alpha}(t,T)
+ \frac{1}{2}\|\bar{b}(t,T)+\bar{\beta}(t,T)\|^2
+\int_t^T\bigl(e^{-g(t,u)}-1\bigr)\xi_t(du)
 = 0.
\]
\end{enumerate}
\end{corollary}

\subsubsection{Generalized Merton models} \label{sec:generalizedMerton}

In the seminal model proposed by R. Merton in \cite{Merton1974}, debt of size $K$ has to be repaid at some (deterministic) date $u_1>0$. Extensions to more sophisticated capital structures have been proposed, amongst others, in \cite{Geske1977,GeskeJohnson84}. In these cases, the credit structure may be incorporated by denoting the dates where obligatory payments are due by $0<u_1<\ldots<u_N$ (such information is often publicly available\footnote{See, for example \url{http://graphics.wsj.com/greece-debt-timeline/} for the debt structure of Greece collected by the Wall Street Journal.}).
Clearly, it is natural to expect discontinuities in the term structure at the dates $\{u_1,\ldots,u_N\}$.
The following corollary deals with this simple setting, to which we refer  as \emph{generalized Merton model} (see also \cite{GehmlichSchmidt2016}).

\begin{corollary}\label{cor:Merton}
Suppose that Assumption \ref{ass:processes} holds and assume furthermore that
\begin{enumerate}[(a)]
\item
the default compensator $H^p$ satisfies $\lambda=0$ and $\{\Delta H^p\neq0\}=\bigcup_{i=1}^{N}\dbra{u_i}$, where $(u_i)_{i=1,\ldots,N}$ is a sequence of deterministic times, for some $N\in\N$;
\item 
$\mu(ds,du)=\sum_{i=1}^{N}\delta_{(0,u_i)}(ds,du)$.
\end{enumerate} 
Then, the probability measure $\Q$ is a local martingale measure for $\{(P(t,T))_{0\leq t\leq T};T\in[0,\T]\}$ with respect to $X^0$ if and only if the following conditions hold a.s.:
\begin{enumerate}[(i)]
\item
$f(t,t)=r_t+h_t$, for Lebesgue-a.e. $t\in[0,\T]$;
\item
$\Delta H^p_{u_i} = 1-e^{-g(u_i,u_i)}$, for all $i\in\N$;
\item
$\bar{a}(t,T)+\bar{\alpha}(t,T) = \frac{1}{2}\|\bar{b}(t,T)+\bar{\beta}(t,T)\|^2$, for all $T\in[0,\T]$ and Lebesgue-a.e. $t\in[0,T]$.
\end{enumerate}
\end{corollary}

In particular, comparing condition {\em (iii)} of Corollary \ref{cor:GS_case1} with condition {\em (iii)} of Corollary \ref{cor:Merton}, we see that there is no compensation for the arrival of news concerning future risky dates. This is simply due to the fact that, under the assumptions of Corollary \ref{cor:Merton}, all risky dates are already publicly known at the initial date $t=0$.

\subsection{General recovery schemes}	\label{sec:gen_rec}

We have so far considered the case where the credit risky bond becomes worthless as soon as the default event occurs. In this section, taking up ideas from \cite{BelangerShreveWong2004,BieleckiRutkowski2002,Duffie1998}, we generalize the above framework to include general recovery schemes, where the credit risky bond is supposed to lose part of its value in correspondence of a sequence of successive credit events. 
Before presenting the general theory, let us consider the following  example.

\begin{example}[Recovery of market value]
Consider an $\bbF$-adapted marked point process $(\tau_n,e_n)_{n\in\N}$, meaning that $(\tau_n)_{n\in\N}$ are stopping times and each random variable $e_n$ is $\cF_{\tau_n}$-measurable. Each stopping time $\tau_n$ denotes a {\em default time} where the credit risky bond loses a fraction $e_n$ of its market value. We assume that the fractional losses $e_n$ take values in $[0,1]$ (with the special case of zero recovery corresponding to $e_n=1$).
Note that, in line with \cite{Schoenbucher98}, the loss at default $e_n$ is possibly unpredictable, but known at the corresponding default time $\tau_n$.
Under this assumption ({\em fractional recovery of market value}), the term structure of credit risky bond prices can be assumed to be of the form
\begin{align*}
P(t,T) = \prod_{\tau_n \le t} \bigl(1-e_n\bigr) \cdot \exp\bigg( -\int_t^T f(t,u) du - \int_{(t,T]} g(t,u) \mu_t(du)\bigg), 
\qquad \text{ for all } 0 \le t \le T \le \T.
\end{align*} 
In this case, we can define the {\em recovery process} $\xi=(\xi_t)_{0\leq t\leq \T}$ by 
\[ 
\xi_t = \prod_{\tau_n \le t} \bigl(1-e_n\bigr), 
\qquad  \text{ for all } 0\leq t\leq \T.
\]
Note that the recovery process $\xi$ is adapted, starts at $\xi_0=1$ and is decreasing.
In this example, strongly inspired by de-facto behavior of bond prices, the recovery process is piecewise constant.
In the following, however, we shall allow for a more general structure. \hfill $\diamond$ 
\end{example}

Inspired by the above example, let us consider a general recovery process $\xi=(\xi_t)_{0\leq t\leq \T}$ satisfying the following assumption. We denote $\tau:=\inf\{t\in[0,\T] : \xi_t=0\}$. 

\begin{assumption}	\label{ass:recovery}
The recovery process $\xi=(\xi_t)_{0\leq t\leq \T}$ is an adapted c\`adl\`ag decreasing non-negative process with $\xi_0=1$ such that $\xi=\xi\Ind_{\dbraco{0,\tau}}$ and $\xi_{\tau-}>0$ a.s.
\end{assumption}

Assumption \ref{ass:recovery} is clearly satisfied by the vast majority of recovery schemes typically considered in practice.
In view of \cite[Theorem 9.41]{HWY}, there exists a c\`adl\`ag decreasing process $R=(R_t)_{0\leq t\leq \T}$ satisfying $-1\leq\Delta R\leq0$ such that $\xi=\cE(R)$.
We denote by $\mu^R$ the jump measure of $R$ and by $\mu^{p,R}$ its compensator.
Since $R$ admits limits from the left and has bounded jumps, it is locally bounded and, hence, special. \cite[Corollary II.2.38]{JacodShiryaev} then implies that the process $R$ admits the canonical representation
\[
R_t = \bigl(x * (\mu^R-\mu^{p,R})\bigr)_t - C_t,
\qquad \text{ for all }0\leq t\leq \T,
\]
where $(C_t)_{0\leq t\leq \T}$ is an increasing predictable process  such that $\Delta C_t=-\int_{[-1,0]}x\mu^{p,R}(\{t\}\times dx)$, for all $t\in[0,\T]$.

Introducing the general recovery process $\xi=\cE(R)$, we extend the term structure \eqref{PtT} as follows:
\begin{align}\label{PtT_2}
  P(t,T) = \cE(R)_t \,\exp\bigg( -\int_t^T f(t,u) du - \int_{(t,T]} g(t,u) \mu_t(du)\bigg), 
\quad\text{ for all } 0 \le t \le T \le \T.
\end{align}
The main goal of the present section consists in obtaining necessary and sufficient conditions for $\Q$ to be a local martingale measure for the family $\{(P(t,T))_{0\leq t\leq T};T\in[0,\T]\}$ with respect to the num\'eraire $X^0=\exp(\int_0^{\cdot}r_tdt)$, thus extending Theorem \ref{thm:dc} to general recovery schemes.

Letting the process $Y^{(T)}$ be defined as in \eqref{eq:YtT}, for every $T\in[0,\T]$, we denote by $\mu^{(Y^{(T)},-R)}$ the jump measure associated to the two-dimensional semimartingale $(Y^{(T)},-R)$, with corresponding compensator $\mu^{p,(Y^{(T)},-R)}$.
We are now in a position to state the following theorem. 
Similarly as above, we use the decomposition $C=\int_0^{\cdot}C^{{\rm ac}}_sds + C^{{\rm sg}} + \sum_{0<s\leq\cdot}\Delta C_s$, with $C^{{\rm ac}}$ and $C^{{\rm sg}}$ denoting respectively the density of the absolutely continuous part and the singular part of $C^c$.

\begin{theorem}	\label{thm:dc2}
Suppose that Assumptions \ref{ass:rnd_meas}, \ref{ass:processes} and \ref{ass:recovery} hold. Let $\Psi$  be defined as in \eqref{eq:def_W}, and $g^{(T)}(\omega;s,u):=\ind{u\leq T}g(\omega;s,u)$, for all $T\in[0,\T]$.
Then, the probability measure $\Q$ is a local martingale measure for $\{(P(t,T))_{0\leq t\leq T};T\in[0,\T]\}$ with respect to $X^0$ if and only if the following conditions hold a.s.:
\begin{enumerate}[(i)]
\item 
$f(t,t) + g(t,t)m_t = r_t + C^{{\rm ac}}_t$ for Lebesgue-a.e. $t\in[0,\T]$;
\item
$\Delta C_t = 1-e^{-g(t,t)\Delta\bar{\mu}_t}$, for all $t\in[0,\T]$;
\item
$\Delta(\Psi * \mu^{p,(Y^{(T)},-R)})_t = 0$, for all $0\leq t\leq T\leq \T$;
\item
for all $T\in[0,\T]$ and for Lebesgue-a.e. $t\in[0,T]$, it holds that
\begin{align*}
&-\bar{a}(t,T)-\bar{\alpha}(t,T)+\frac{1}{2}\bigl\|\bar{b}(t,T)+\bar{\beta}(t,T)\bigr\|^2-(g^{(T)} * \mu)^{{\rm ac}}_t	+ (\Psi * \mu^{p,(Y^{(T)},-R)})^{{\rm ac}}_t = 0;
\end{align*}
\item
for all $0\leq t\leq T\leq \T$, it holds that
\[
\int_0^tg(s,s)d\nu_s-(g^{(T)} * \mu)^{{\rm sg}}_t+(\Psi * \mu^{p,(Y^{(T)},-R)})^{{\rm sg}}_t = C^{{\rm sg}}_t.
\]
\end{enumerate}
\end{theorem}

The interpretation of the five necessary and sufficient conditions stated in the above theorem is analogous to the case of Theorem \ref{thm:dc}.

\subsection{Related literature}\label{sec:rel-literature}

As mentioned in the introduction, the two classical approaches to credit risk modelling are the \emph{structural} approach, starting with Merton \cite{Merton1974} and its extensions \cite{Geske1977,GeskeJohnson84}, and the \emph{reduced-form} approach, introduced in early works of Jarrow, Lando and Turnbull \cite{JarrowTurnbull1995,Lando94} and in Artzner and Delbaen \cite{ArtznerDelbaen95}. 
It was a long time that these approaches co-existed in the literature but no model was bridging them (see however  \cite{DuffieLando2001, JarrowProtter, FreySchmidt2009, FreySchmidt2012} for information-based models connecting structural and reduced-form models). 
In more recent years, some of the features of structural and reduced-form models have been combined in {\em hybrid} models, as considered for instance in \cite{CPS,CarrLinetsky,MCL}. In particular, in hybrid models the default compensator has an absolutely continuous part and a discontinuous part, thus showing once more the importance of considering the possibility that default occurs in correspondence of predictable times. 
The model considered in \cite{Andreasen} also shows a similar behavior.
However, no general theory of term structure modelling was available for hybrid models so far. The present paper intends to fill this gap.
In the remaining part of this section, we discuss in detail the relation with the works \cite{BelangerShreveWong2004} and \cite{JiaoLi2015,JiaoLi2016}, which are especially related to our framework.


\subsubsection{The relation to B\'elanger, Shreve and Wong (2004)} 

The remarkable paper \cite{BelangerShreveWong2004} considers a first-passage-time model over a random boundary for the default time and points towards an extension of the reduced-form approach beyond intensity-based models. The framework may be seen as a structural approach where the debt level is random and we give a short account. 
In \cite{BelangerShreveWong2004}, the authors consider a filtration $\bbG$, given by the augmented filtration generated by a Brownian motion $W$. Additionally, there is a c\`adl\`ag non-decreasing $\bbG$-predictable process $(\Lambda_t)_{0\leq t\leq \T}$ and the default time $\tau$ is defined as
$$ \tau := \inf\bigl\{t \in[0,\T]: \Lambda_t \ge \Theta \bigr\}, $$
where $\Theta$ is a strictly positive random variable independent of $\bbG$. The filtration $\bbF$ is then defined as the progressive enlargement of $\bbG$ with respect to $\tau$.
Depending on the choice of the process $(\Lambda_t)_{0\leq t\leq \T}$, it is shown that the default compensator $H^p$ may contain jumps as well as a singular continuous part, thus exploiting the generality of  decomposition \eqref{eq:dec_comp} and going beyond classical intensity-based models. 
However, the HJM approach to the modelling of defaultable term structures is only considered  in \cite[Section 5]{BelangerShreveWong2004}  in an intensity-based setting (i.e., assuming that the default compensator $H^p$ is absolutely continuous). 

\subsubsection{The relation to Jiao and Li (2015)}

More recently, extensions of the intensity-based approach have been pursued via methods of enlargements of filtrations, see \cite{ElKarouiJeanblancJiao2010,Karoui2014conditional}. This approach has been extended in \cite{JiaoLi2015} to the case where the default compensator exhibits discontinuities. 
Starting from a background filtration $\bbG$, \cite{JiaoLi2015} consider a finite family $\{\tau_1,\ldots,\tau_n\}$ of $\bbG$-stopping times, which can be chosen strictly increasing without loss of generality. The filtration $\bbF$ is constructed as the progressive enlargement of $\bbG$ with respect to the default time $\tau$.
Letting $(\alpha_t)_{0\leq t\leq \T}$ be a $\bbG$-optional process taking values in the space of measurable functions on $\R_+$, \cite{JiaoLi2015} propose the following \emph{generalized density hypothesis}:
\[
\E\Bigg[ \ind{\tau < +\infty} h(\tau) \prod_{i=1}^n \ind{\tau \neq \tau_i}\biggr|\cG_t\Bigg] = \int_{\R_+} h(u) \alpha_t(u) \eta(du) \qquad \text{a.s. for all }0\leq t\leq \T,  
\]
for any bounded measurable function $h(\cdot)$, where $\eta$ is assumed to be a non-atomic $\sigma$-finite Borel measure on $\R_+$. 
In \cite[Section 3]{JiaoLi2015}, the default compensator $H^p$ is computed under this hypothesis. It is shown that $H^p$ contains an integral with respect to the measure $\eta$, which may not necessarily be absolutely continuous and, in addition, $H^p$ depends on the $\bbG$-compensators of the processes $\Ind_{\dbraco{\tau_i,+\infty}}$ which are allowed to be fully general and may therefore exhibit a jumping behavior. Clearly, this specification can be covered by the general decomposition \eqref{eq:dec_comp}.

As an example, \cite{JiaoLi2015} consider the case $\tau:=\tau_1 \wedge E$, where $E$ is exponentially distributed and $\tau_1$ is the first passage time of a Brownian motion at the level $a<0$. In this case it, follows that 
\begin{align} \label{eq:Jiao1} 
H^p_t = \int_0^{t} h_s ds + \ind{\tau_1 \le t} \Delta H_{\tau_1}^p, 
\qquad \text{ for all }0 \le t \le \tau.
\end{align}
Our results can be applied to this setting and permit to describe the general class of arbitrage-free term structure models compatible with this structure of the default compensator.

The approach of \cite{JiaoLi2015} has been recently extended to the context of sovereign default risk in \cite{JiaoLi2016}. Consider a sequence of increasing thresholds $0<a_1<a_2<\ldots<a_n$ and denote by $(\tau_i)_{i=1,\ldots,n}$ the (increasing) first-passage times of a Brownian motion of these levels. Let $E'$ be an independent exponentially distributed random variable. The times $\tau_i$ represent critical political events where the sovereign seeks financial aid to avoid immediate default. If $\tau_i > E'$, this attempt was not successful and default occurs. Furthermore, \cite{JiaoLi2016} consider an additional doubly-stochastic random time with an intensity and let the default time $\tau$ be the minimum of such times. In summary, the authors show that $H^p_t = \int_0^th_s ds + \sum_{i=1}^n \ind{\tau_i \le t} \Delta H_{\tau_i}^p$, for all $0 \le t \le \tau$. The authors also study the case where the Brownian motion is replaced by a diffusive Markov process and obtain explicit formulas for  geometric Brownian motion and for the CEV process.

\section{Proofs}
\label{sec:proof}

This section contains the proofs of all our results. After giving the proof of  two technical lemmata stated in Section \ref{sec:tau}, we present the proof of Theorem \ref{thm:dc}, while the proofs of the results stated in Section \ref{sec:examples} are given in Section \ref{sec:proofs_cor}. Finally, Section \ref{sec:proof_recovery} contains the proof of Theorem \ref{thm:dc2}.

\subsection{Proofs of the results of Section~\ref{sec:tau}}

\begin{proof}[Proof of Lemma~\ref{lem:dec_comp}]
Since $H^p$ is a predictable process of finite variation, it can be decomposed as $H^p=(H^p)^c+\sum_{0<s\leq \cdot}\Delta H^p_s$, where $(H^p)^c$ is an increasing continuous process. 
\cite[Theorem 2.1]{DelbaenSchachermayer1995} then yields the existence of an integrable predictable process $(h_t)_{0\leq t\leq \T}$ such that
\[ 
(H^p)^c_t = \int_0^t h_s ds + \int_0^t \Ind_N(s) d(H^p)^c_s,
\qquad \text{ for all }0\leq t\leq \T,
\]
where $N$ is a predictable subset of $\Omega\times[0,\T]$ such that the sections $N_{\omega}:=\{t\in[0,\T]:(\omega,t)\in N\}$ have Lebesgue measure zero, for a.a. $\omega\in\Omega$.
The result  follows by letting $\lambda:=\int_0^{\cdot}\Ind_N(s) d(H^p)^c_s$.
\end{proof}

\begin{proof}[Proof of Lemma~\ref{lem:predictability}]
Note first that, due to Assumption \eqref{ass:rnd_meas}, it holds that, for every $t\in[0,\T]$,
\[
\bar{\mu}_t 
= \mu\bigl([0,\T]\times[0,t]\bigr)
= \mu\bigl([0,t]\times[0,t]\bigr)
= \int_0^t\int_0^{\T}\ind{u\leq t}\mu(ds,du).
\]
For any measurable bounded function $\vartheta:[0,\T]\times[0,\T]\rightarrow\R_+$, the process $(\Theta(t,v))_{0\leq t\leq \T}$ defined by $\Theta(t,v):=\int_0^t\int_0^{\T}\vartheta(v,u)\mu(ds,du)$ is optional and increasing, for every $v\in[0,\T]$, since $\mu(ds,du)$ is a non-negative optional random measure. Being increasing, $(\Theta(t,v))_{0\leq t\leq \T}$ admits limits from the left, so that the process $(\Theta(t-,v))_{0\leq t\leq \T}$ is adapted and left-continuous, hence predictable, for every $v\in[0,\T]$. 
In turn, this implies that the processes $(\Theta(t,t))_{0\leq t\leq \T}$ and $(\Theta(t-,t))_{0\leq t\leq \T}$ are optional and predictable, respectively. Indeed, this is obvious for functions of the form $\vartheta(t,u)=p(t)q(u)$, with $p,q:[0,\T]\rightarrow\R_+$ bounded and measurable 
and the general case follows by a monotone class argument.
Letting $\vartheta(t,u)=\ind{u\leq t}$, this shows that the process $(\bar{\mu}_t)_{0\leq t\leq \T}$ is optional and increasing. Moreover, due to part (i) of Assumption \eqref{ass:rnd_meas}, it holds that 
\[
\bar{\mu}_t = \mu\bigl([0,t]\times[0,t]\bigr) = \mu\bigl([0,t)\times[0,t]\bigr)=\Theta(t-,t),
\] 
thus showing the predictability of $\bar{\mu}$.
The right-continuity of $\bar{\mu}$ follows by the upper semicontinuity of the measure. Decomposition \eqref{eq:dec_pred_proc} can be obtained by the same arguments used in the proof of Lemma~\ref{lem:dec_comp}.
\end{proof}

\subsection{Proof of Theorem~\ref{thm:dc}}	\label{sec:proof_thm}

Since the proof of Theorem~\ref{thm:dc} requires several intermediate steps, let us first give an outline of the main ideas involved. The starting point consists in representing the pre-default price (i.e., on the set $\{H=0\}$) of a credit risky bond as an exponential of a semimartingale admitting an explicit decomposition into a predictable finite variation part, a continuous local martingale part and an integral with respect to the random measure $\mu$. As a second step, we conveniently transform the ordinary exponential into a stochastic exponential. The desired local martingale property of (discounted) credit risky bond prices will then be equivalent to the local martingale property of the process defining the stochastic exponential. By computing the canonical decomposition of the latter, the local martingale property will hold if and only if all predictable finite variation terms vanish. This will lead to the conditions stated in Theorem \ref{thm:dc}. 

We start by rewriting the defaultable bond price $P(t,T)$ in the following form:
\begin{equation}	\label{eq:bond_prod}
P(t,T) = (1-H_t)F(t,T)G(t,T),
\end{equation}
where 
\[
F(t,T):=\exp\left(-\int_t^Tf(t,u)du\right)
\qquad\text{ and }\qquad
G(t,T):=\exp\left(-\int_{(t,T]}g(t,u)\mu_t(du)\right),
\]
for all $0\leq t\leq T\leq \T$.
By Assumption \ref{ass:processes} and following the original arguments of \cite{HJM}, the term $F(t,T)$ admits the representation
\begin{equation}	\label{eq:F}
F(t,T) = \exp\left(\int_0^tf(s,s)ds-\int_0^t\bar{a}(s,T)ds-\int_0^t\bar{b}(s,T)dW_s\right),
\end{equation}
see, e.g., \cite[Lemma 6.1]{Filipovic2009} (note that, in comparison to this work, we rely on a slightly weaker assumption on the volatility process $b$ for the application of the stochastic Fubini theorem by virtue of  \cite[Theorem IV.65]{Protter} or \cite[Proposition A.2]{BMKR}). 

The next lemma, which extends \cite[Lemma 2.3]{GehmlichSchmidt2016MF} to the present general setting, derives a representation analogous to \eqref{eq:F} for the term $G(t,T)$.

\begin{lemma}	\label{lem:compute_G}
Suppose that Assumptions \ref{ass:rnd_meas} and \ref{ass:processes} hold.  
Then, for each $T\in[0,\T]$, the process $(\log G(t,T))_{0\leq t\leq T}$ is a semimartingale admitting the decomposition
\begin{equation}	\label{eq:semimg_repr_G}
\log G(t,T) = \int_0^tg(s,s)\,d\bar{\mu}_s - \int_0^t\bar{\alpha}(s,T)ds - \int_0^t\bar{\beta}(s,T)dW_s -  \int_0^t\int_{(s,T]} g(s,u)\mu(ds,du).
\end{equation}
\end{lemma}
\begin{proof}
We first show that the stochastic integral $\int_0^{\cdot}\bar{\beta}(s,T)dW_s$ is well-defined, for every $T\in[0,\T]$. To this effect,
H\"older's inequality and Assumptions \ref{ass:rnd_meas} and \ref{ass:processes} imply that, for every $T\in[0,\T]$,
\begin{align*}
\int_0^T\|\bar{\beta}(s,T)\|^2ds
&= \int_0^T\biggl\|\int_{(s,T]}\beta(s,u)\mu_s(du)\biggr\|^2ds
\leq \int_0^T\left(\mu_s([0,T])\int_0^T\|\beta(s,u)\|^2\mu_s(du)\right)ds	\\
& \leq \mu_T([0,T])\int_0^T\int_0^T\|\beta(s,u)\|^2\mu_s(du)ds < +\infty\text{ a.s.}
\end{align*}
thus proving the well-posedness of the stochastic integral $\int_0^{\cdot}\bar{\beta}(s,T)dW_s$.
In turn, since the term $G(t,T)$ is a.s. finite for every $0\leq t\leq T\leq \T$, 
the decomposition \eqref{eq:semimg_repr_G} implies that the term $\int_0^t\int_{(s,T]}g(s,u)\mu(ds,du)=Y_t^{(T)}$ is a.s. finite and well-defined as a finite variation process.

Observe that, by the definition of $\mu_t(du)$,
\begin{equation}	\label{eq:dyn_1}
-\log G(t,T)
= \int_0^t\int_{(t,T]}g(t,u)\mu(ds,du)
= \int_0^t\int_0^T\ind{u>t}g(t,u)\mu(ds,du).
\end{equation}
The product rule, together with equation \eqref{g1} and the continuity of $g$, yields that
\begin{align}
\Ind_{[0,u)}(t)g(t,u)
&= g(0,u) + \int_0^t \Ind_{[0,u)}(v)\, dg(v,u)+ \int_0^t g(v,u)d\bigl(\Ind_{[0,u)}(v)\bigr) \notag\\
&= g(0,u)  + \int_0^t\Ind_{[0,u)}(v)\alpha(v,u)dv + \int_0^t\Ind_{[0,u)}(v)\beta(v,u) dW_v - g(u,u)\ind{u\leq t},
\label{eq:dyn_2}
\end{align}
where both integrals are well-defined by Assumption \ref{ass:processes}.
Equations \eqref{eq:dyn_1}-\eqref{eq:dyn_2} imply that
\begin{align}
-\log G(t,T) &= \int_0^t\int_0^T g(0,u) \mu(ds,du) + \int_0^t\int_0^T \int_0^t\Ind_{[0,u)}(v)\alpha(v,u)dv \mu(ds,du) \notag\\
& + \int_0^t\int_0^T \int_0^t\Ind_{[0,u)}(v)\beta(v,u)dW_v\, \mu(ds,du) - \int_0^t\int_0^T g(u,u)\ind{u\leq t}\mu(ds,du) \notag\\ & =: (1) + (2) + (3)+(4).
\label{eq:dyn_3}
\end{align}
Due to part (ii) of Assumption \ref{ass:processes}, we can apply for each $\omega\in\Omega$ the classical Fubini theorem to the term $(2)$, so that
\begin{align*}
(2) &= 
\int_0^t\int_0^T \int_0^s\Ind_{[0,u)}(v)\alpha(v,u)dv \mu(ds,du)
+ \int_0^t\int_0^T \int_s^t\Ind_{[0,u)}(v)\alpha(v,u)dv \mu(ds,du)	\\
&= \int_0^t\int_0^T \int_0^s\Ind_{[0,u)}(v)\alpha(v,u) \, dv \mu(ds,du)
+ \int_0^t \int_0^v \int_0^T\Ind_{[0,u)}(v)\alpha(v,u) \mu(ds,du)\, dv.
\end{align*}
As shown in Lemma \ref{lem:stoch_Fubini} in the appendix, Assumptions \ref{ass:rnd_meas} and \ref{ass:processes} allow to perform an analogous change of the order of integration in the term (3), so that
\[
(3) = \int_0^t\int_0^T \int_0^s\Ind_{[0,u)}(v)\beta(v,u) \, dW_v \mu(ds,du)
+ \int_0^t \int_0^v \int_0^T\Ind_{[0,u)}(v)\beta(v,u) \mu(ds,du)\, dW_v.
\]
Note also that
\begin{align*}
\int_0^t \int_0^v \int_0^T\Ind_{[0,u)}(v)\alpha(v,u) \mu(ds,du)\, dv 
&= \int_0^t \int_{(v,T]}\alpha(v,u) \mu_v(du)\, dv
= \int_0^t \bar \alpha(v,T) dv,	\\
\int_0^t \int_0^v \int_0^T\Ind_{[0,u)}(v)\beta(v,u) \mu(ds,du)\, dW_v 
&= \int_0^t \int_{(v,T]}\beta(v,u) \mu_v(du)\, dW_v
= \int_0^t \bar \beta(v,T) dW_v,
\end{align*}
where both integrals are well-defined by Assumption \ref{ass:processes}.
Moreover, due to equation \eqref{eq:dyn_2}, it holds that
\[
\int_0^s\Ind_{[0,u)}(v)\alpha(v,u)  dv + \int_0^s\Ind_{[0,u)}(v)\beta(v,u) dW_v  
= \Ind_{[0,u)}(s)g(s,u) - g(0,u) +  g(u,u) \ind{u \le s},
\]
so that equation \eqref{eq:dyn_3} can be rewritten as
\begin{align*}
-\log G(t,T)
&= \int_0^t \bar \alpha(v,T) dv + \int_0^t \bar \beta(v,T) dW_v	\\
&\quad+ \int_0^t\int_0^T \Ind_{[0,u)}(s) g(s,u)\,\mu(ds,du) -  \int_0^t\int_0^T \Ind_{(s,t]}(u) g(u,u) \mu(ds,du).
\end{align*}
Finally, part (i) of Assumption \ref{ass:rnd_meas} and the definition of the process $\bar{\mu}$ imply that
\[
\int_0^t\int_0^T \Ind_{(s,t]}(u) g(u,u) \mu(ds,du) 
= \int_0^t\int_0^t  g(u,u) \mu(ds,du)	
= \int_0^{\T}\int_0^t  g(u,u) \mu(ds,du)
= \int_0^t g(u,u)\,d\bar{\mu}_u.
\]
As already remarked at the end of Section \ref{sec:prelim}, the process $\int_0^{\cdot}g(u,u)\,d\bar{\mu}_u$ is predictable and of finite variation. This implies the semimartingale property of the process $(\log G(t,T))_{0\leq t\leq T}$, for every $T\in[0,\T]$.
\end{proof}

For each $T\in[0,\T]$, let us define the process $(X^{(T)}_t)_{0\leq t\leq T}$ by
\begin{equation}	\label{eq:XtT}	\begin{aligned}
X^{(T)}_t &:= 
\log\bigl(F(t,T)\bigr) + \log\bigl(G(t,T)\bigr)	\\
&= \int_0^tf(s,s)ds-\int_0^t\bar{a}(s,T)ds-\int_0^t\bar{b}(s,T)dW_s	\\
&\quad+\int_0^tg(s,s)d\bar{\mu}_s - \int_0^t\bar{\alpha}(s,T)ds - \int_0^t\bar{\beta}(s,T)dW_s -  \int_0^t\int_{(s,T]} g(s,u)\mu(ds,du),
\end{aligned}	\end{equation}
so that $P(t,T)=(1-H_t)\exp(X^{(T)}_t)$.
In the following lemma, we give an alternative representation of the defaultable bond price $P(t,T)$ as a stochastic exponential.

\begin{lemma}	\label{lem:stoch_exp}
Suppose that Assumptions \ref{ass:rnd_meas} and \ref{ass:processes} hold. 
Then, for each $0\leq t\leq T\leq \T$, the credit risky bond price $P(t,T)$ can be represented as
\begin{equation}	\label{eq:stoch_exp_1}
P(t,T) = \mathcal{E}\bigl(\widetilde{X}^{(T)}-H-[\widetilde{X}^{(T)},H]\bigr)_t,
\end{equation}
where, for each $T\in[0,\T]$, the process $(\widetilde{X}^{(T)}_t)_{0\leq t\leq T}$ is defined as
\begin{align}
\widetilde{X}^{(T)}_t
&:= X^{(T)}_t + \frac{1}{2}\int_0^t\|\bar{b}(s,T)+\bar{\beta}(s,T)\|^2ds	\notag\\
&\quad + \sum_{0<s\leq t}\left(e^{-\int_s^Tg(s,u)\mu(\{s\}\times du)+g(s,s)\Delta\bar{\mu}_s}-1+\int_s^Tg(s,u)\mu(\{s\}\times du)-g(s,s)\Delta\bar{\mu}_s\right).
\label{eq:stoch_exp_2}
\end{align}
\end{lemma}
\begin{proof}
Since the process $(H_t)_{0\leq t\leq \T}$ is a single jump process with jump size equal to one, it follows that, by the definition of stochastic exponential,
\[
1-H_t = e^{H_0-H_t}\prod_{0<s\leq t}(1-\Delta H_s)e^{\Delta H_s} = \mathcal{E}(-H)_t,
\qquad\text{ for all }t\in[0,\T].
\]
\cite[Theorem II.8.10]{JacodShiryaev} implies that $\exp(X^{(T)}_t)=\mathcal{E}(\widetilde{X}^{(T)})_t$, for all $0\leq t\leq T\leq \T$, where  $\widetilde{X}^{(T)}$ is defined as in \eqref{eq:stoch_exp_2}. Representation \eqref{eq:stoch_exp_1} then follows by Yor's formula (see \cite[\textsection~II.8.19]{JacodShiryaev}).
\end{proof}

Our next goal consists in developing a more tractable representation of the process appearing in the stochastic exponential in \eqref{eq:stoch_exp_1}. To this effect, let us analyze in more detail the jumps of the semimartingale $\widetilde{X}^{(T)}-H-[\widetilde{X}^{(T)},H]$:
\begin{align}	
& \Delta\bigl(\widetilde{X}^{(T)}-H-[\widetilde{X}^{(T)},H]\bigr)_t
= \Delta\widetilde{X}^{(T)}_t-\Delta H_t-\Delta H_t\Delta\widetilde{X}^{(T)}_t	\notag\\
&\quad =  
e^{-\int_t^Tg(t,u)\mu(\{t\}\times du)+g(t,t)\Delta\bar{\mu}_t}-1 -\Delta H_t	 - \Delta H_t\left(e^{-\int_t^Tg(t,u)\mu(\{t\}\times du)+g(t,t)\Delta\bar{\mu}_t}-1\right)	\notag\\
&\quad =
e^{g(t,t)\Delta\bar{\mu}_t}\left(e^{-\int_t^Tg(t,u)\mu(\{t\}\times du)}-1\right)(1-\Delta H_t)
+ \left(e^{g(t,t)\Delta\bar{\mu}_t}-1\right)(1-\Delta H_t)-\Delta H_t.	\label{eq:733}
\end{align}
Let us rewrite this last expression in a more compact way by using the notation introduced in Section \ref{sec:main_result}. 
To this end, for each $T\in[0,\T]$, we make use of the processes $Y^{(T)}=\int_0^{\cdot}\int_0^Tg(s,u)\mu(ds,du)$ and $H$ and of the corresponding jump measure $\mu^{(Y^{(T)},H)}$, so that
\begin{equation}	\label{eq:W1_jumps}
\sum_{0<s\leq t}e^{g(s,s)\Delta\bar{\mu}_s}\left(e^{-\int_s^Tg(s,u)\mu(\{s\}\times du)}-1\right)(1-\Delta H_s)
= (\Psi * \mu^{(Y^{(T)},H)})_t,
\end{equation}
with $\Psi(\omega;s,y,z) = e^{g(\omega;s,s)\Delta\bar{\mu}_s(\omega)}\left(e^{-y}-1\right)(1-z)$. Note that the function $\Psi$ is $\mathcal{P}\otimes\mathcal{B}(\mathbb{R}\times[0,1])$-measurable and $\Psi * \mu^{(Y^{(T)},H)}$ makes sense as an integral with respect to the random measure $\mu^{(Y^{(T)},H)}$. Indeed, \eqref{eq:W1_jumps} is well-defined, since $H$ is a single jump process and
\begin{align}
& \sum_{0<s\leq t}e^{g(s,s)\Delta\bar{\mu}_s}\left|e^{-\int_s^Tg(s,u)\mu(\{s\}\times du)}-1\right|\notag\\
&
\leq \sum_{0<s\leq t}\left|e^{-\int_s^Tg(s,u)\mu(\{s\}\times du)+g(s,s)\Delta\bar{\mu}_s}-1\right|
+ \sum_{0<s\leq t}\left|e^{g(s,s)\Delta\bar{\mu}_s}-1\right|
<+\infty \text{ a.s. }
\label{eq:ineq_jumps}
\end{align}
In fact, the first sum appearing in \eqref{eq:ineq_jumps} is finite as a consequence of \eqref{eq:stoch_exp_2} together with the fact that the two processes $Y^{(T)}$ and $\int_0^{\cdot}g(s,s)d\bar{\mu}_s$ are of finite variation, which in turn implies that $\sum_{0<s\leq{\cdot}}\int_s^Tg(s,u)\mu(\{s\}\times du)$ and $\sum_{0<s\leq\cdot}g(s,s)\Delta\bar{\mu}_s$ are a.s. finite. Moreover, the process $\int_0^{\cdot}g(s,s)d\bar{\mu}_s$ is predictable and of finite variation, hence locally bounded (see  \cite[Lemma I.3.10]{JacodShiryaev}) and exponentially special, so that the second term appearing in \eqref{eq:ineq_jumps} is a.s. finite by \cite[Proposition II.8.26]{JacodShiryaev}.
This shows that the summations of the jump terms appearing in \eqref{eq:733} are well-defined and a.s. finite.
We have thus obtained the representation:
\[
\sum_{0<s\leq t}\Delta\bigl(\widetilde{X}^{(T)}-H-[\widetilde{X}^{(T)},H]\bigr)_s
= \sum_{0<s\leq t}\left(e^{g(s,s)\Delta\bar{\mu}_s}-1\right)(1-\Delta H_s) - H_t
+(\Psi * \mu^{(Y^{(T)},H)})_t.
\]
In turn, together with the definition of the process $\widetilde{X}^{(T)}$ (see \eqref{eq:stoch_exp_2}) and decomposition \eqref{eq:dec_pred_proc}, this implies that the semimartingale $\widetilde{X}^{(T)}-H-[\widetilde{X}^{(T)},H]$ defining the stochastic exponential \eqref{eq:stoch_exp_1} admits the following decomposition:
\begin{align}
\widetilde{X}^{(T)}_t-H_t-[\widetilde{X}^{(T)},H]_t & = 
\int_0^tf(s,s)ds
-\int_0^t\bar{a}(s,T)ds
- \int_0^t\bar{\alpha}(s,T)ds 
\notag\\
&\quad 
+ \frac{1}{2}\int_0^t\|\bar{b}(s,T)+\bar{\beta}(s,T)\|^2ds
+ \int_0^tg(s,s)m_sds 
+ \int_0^tg(s,s)d\nu_s	\notag\\
&\quad 
-\int_0^t\bar{b}(s,T)dW_s
- \int_0^t\bar{\beta}(s,T)dW_s 
-(g^{(T)} * \mu)^c_t  
\label{eq:stoch_exp_dec}\\
&\quad 
+ \sum_{0<s\leq t}\bigl(e^{g(s,s)\Delta\bar{\mu}_s}-1\bigr)(1-\Delta H_s)
-H_t
+(\Psi * \mu^{(Y^{(T)},H)})_t,
\notag
\end{align}
where $g^{(T)}(\omega;s,u):=\ind{u\leq T}g(\omega;s,u)$ and $(g^{(T)} * \mu)^c$ denotes the continuous part of the finite variation process $g^{(T)} * \mu=Y^{(T)}$.

We are now in a position to complete the proof Theorem \ref{thm:dc}.

\begin{proof}[Proof of Theorem \ref{thm:dc}]
Recall that, in view of Lemma \ref{lem:predictability}, the process $(g(t,t)\Delta\bar{\mu}_t)_{0\leq t\leq \T}$ is predictable and locally bounded, since $(g(t,t))_{0\leq t\leq \T}$ is continuous and $(\bar{\mu}_t)_{0\leq t\leq \T}$ is locally bounded, being a predictable process of finite variation. Hence, by compensating the process $H$ and using decomposition \eqref{eq:dec_comp}, it follows that
\begin{align*}
\sum_{0<s\leq t}\bigl(e^{g(s,s)\Delta\bar{\mu}_s}-1\bigr)\Delta H_s
&= \int_0^t\bigl(e^{g(s,s)\Delta\bar{\mu}_s}-1\bigr)dH_s	\\
&= \text{ (local martingale)}_t
+ \int_0^t\bigl(e^{g(s,s)\Delta\bar{\mu}_s}-1\bigr)dH^p_s	\\
&= \text{ (local martingale)}_t
+ \sum_{0<s\leq t}\bigl(e^{g(s,s)\Delta\bar{\mu}_s}-1\bigr)\Delta H^p_s.
\end{align*}
Recall that, for each $T\in[0,\T]$, the random measure $\mu^{p,(Y^{(T)},H)}$ denotes the compensator of $\mu^{(Y^{(T)},H)}$, in the sense of \cite[Theorem II.1.8]{JacodShiryaev}.
Hence, by relying on \eqref{eq:stoch_exp_dec} together with Lemma \ref{lem:dec_comp}, we obtain that
\begin{equation}	\label{eq:FV}	\begin{aligned}
\widetilde{X}^{(T)}_t-H_t-[\widetilde{X}^{(T)},H]_t & = 
\text{ (local martingale)}_t	\\
&\quad 
+ \int_0^tf(s,s)ds
-\int_0^t\bar{a}(s,T)ds
- \int_0^t\bar{\alpha}(s,T)ds \\
&\quad 
+ \frac{1}{2}\int_0^t\|\bar{b}(s,T)+\bar{\beta}(s,T)\|^2ds
+ \int_0^tg(s,s)m_sds
+ \int_0^tg(s,s)d\nu_s	\\
&\quad
- (g^{(T)} * \mu)^c_t
- \int_0^t h_s ds 
- \lambda_t \\
&\quad 
- \sum_{0<s\leq t}\Delta H^p_s
+ \sum_{0<s\leq t}(e^{g(s,s)\Delta\bar{\mu}_s}-1)(1-\Delta H^p_s)	\\
&\quad
+(\Psi * \mu^{p,(Y^{(T)},H)})_t,
\end{aligned}	\end{equation}
where we have used the fact that the finite variation process $(g^{(T)} * \mu)^c$ is predictable, being adapted and continuous.
Taking into account equation \eqref{eq:stoch_exp_1} and by \cite[Corollary I.3.16]{JacodShiryaev}, this implies that the discounted defaultable bond price $(P(t,T)/X^0_t)_{0\leq t\leq T}$ is a local martingale, for every $T\in[0,\T]$, if and only if the predictable finite variation part in \eqref{eq:FV} coincides with $\int_0^{\cdot}r_sds$. 
To this effect, let us 
analyze separately the absolutely continuous, singular and jump parts, for all $0\leq t\leq T\leq\T$. Beginning with the jump terms, it must hold that
\begin{equation}	\label{eq:proof_jumps}
-\Delta H^p_t + (e^{g(t,t)\Delta\bar{\mu}_t}-1)(1-\Delta H^p_t)
+ \Delta(\Psi * \mu^{p,(Y^{(T)},H)})_t = 0.
\end{equation}
Let $t=T$ and note that, in view of \cite[Proposition II.1.17]{JacodShiryaev}, it holds that
\[
\Delta(\Psi * \mu^{p,(Y^{(t)},H)})_t = \E\bigl[\Delta(\Psi * \mu^{(Y^{(t)},H)})_t \bigr|\cF_{t-}\bigr]
= e^{g(t,t)\Delta\bar{\mu}_t}\E\bigl[\bigl(e^{-\Delta Y^{(t)}_t}-1\bigr)(1-\Delta H_t)\bigr|\cF_{t-}\bigr] = 0,
\]
since $\Delta Y^{(t)}_t=0$, for all $t\in[0,T]$. 
Therefore, 
\[
-\Delta H^p_t + (e^{g(t,t)\Delta\bar{\mu}_t}-1)(1-\Delta H^p_t) = 0,
\]
which corresponds to condition {\em (ii)} of Theorem \ref{thm:dc}. In view of \eqref{eq:proof_jumps}, condition {\em (iii)} also follows.
Considering now the continuous singular terms of the finite variation part of  \eqref{eq:FV}, it must hold that
\[
 \int_0^tg(s,s)d\nu_s
- (g^{(T)} * \mu)^{{\rm sg}}_t
- \lambda_t 
+(\Psi * \mu^{p,(Y^{(T)},H)})^{{\rm sg}}_t = 0,
\]
for all $0\leq t\leq T\leq \T$, which yields condition {\em (v)}.
Finally, considering the densities of the absolutely continuous terms of the finite variation part of \eqref{eq:FV} and letting $t=T$, it must hold that
\begin{equation}	\label{eq:cond_short}
f(t,t)  + g(t,t)m_t - (g^{(t)} * \mu)^{{\rm ac}}_t
- h_t  + (\Psi * \mu^{p,(Y^{(t)},H)})^{{\rm ac}}_t = r_t,
\end{equation}
for Lebesgue-a.e.  $t\in[0,\T]$.
However, denoting by $A^{(T),{\rm ac}}$ the density of the absolutely continuous part of the predictable integrable process $A^{(T)}$ appearing in \eqref{eq:kernel1}, for $T\in[0,\T]$, it holds that 
\[
(\Psi * \mu^{p,(Y^{(t)},H)})^{{\rm ac}}_t = A^{(t),{\rm ac}}_t\int_{\R\times\{0,1\}}(e^{-y}-1)(1-z)K^{(t)}(t;dy,dz) = 0,
\]
since $\ind{y\neq0}K^{(t)}(t;dy,dz)=0$ for all $t\in[0,\T]$.
Moreover, by the same arguments used in the proof of \cite[Theorem II.1.8]{JacodShiryaev} (but with respect to the optional $\sigma$-field), it can be shown that there exist a kernel $N(\omega,t;du)$ from $(\Omega\times[0,\T],\cO)$ into $([0,\T],\cB([0,\T]))$ and a predictable integrable increasing process $D=(D_t)_{0\leq t\leq \T}$ such that $\mu(\omega;ds,du)=N(\omega,t;du)dD_s(\omega)$.
Letting $D^c=\int_0^{\cdot}D^{{\rm ac}}_sds+D^{{\rm sg}}$ be the decomposition of the continuous part of the process $D$ into an absolutely continuous part and a singular continuous part, it then follows that, for all $0\leq t\leq T\leq \T$,
\[
(g^{(T)} * \mu)^c_t
= \int_0^t\int_{(s,T]}g(s,u)N(s;du)D^{{\rm ac}}_sds
+ \int_0^t\int_{(s,T]}g(s,u)N(s;du)dD^{{\rm sg}}_s,
\]
so that $(g^{(t)} * \mu)_t^{{\rm ac}}=0$, for all $t\in[0,\T]$.
Therefore, we have shown that condition \eqref{eq:cond_short} reduces to
\[
f(t,t)  + g(t,t)m_t - h_t  = r_t,
\]
for Lebesgue-a.e. $t\in[0,\T]$, which corresponds to condition {\em (i)} in Theorem \ref{thm:dc}. 
Condition {\em (iv)} then follows by considering the remaining absolutely continuous terms and making use of condition {\em (i)}.
Conversely, it can be easily checked that conditions {\em (i)}-{\em (v)} of Theorem \ref{thm:dc} together imply that all the finite variation terms appearing in \eqref{eq:FV} vanish.
\end{proof}

The following simple lemma proves relation \eqref{eq:link_comp_meas}.

\begin{lemma}	\label{lem:rel_comp_meas}
Suppose that Assumptions \ref{ass:rnd_meas} and \ref{ass:processes} hold. 
Then, for every $T\in[0,\T]$, the compensating measure $\mu^{p,(Y^{(T)},H)}$ is related to the compensating measure $\mu^p$ as follows:
\[
\int_{\mathbb{R}}y\,\mu^{p,(Y^{(T)},H)}(\{t\}\times dy\times\{0,1\})
= \int_{(t,T]}g(t,u)\,\mu^p(\{t\}\times du),
\qquad\text{ for all } 0\leq t \leq T.
\]
\end{lemma}
\begin{proof}
It suffices to remark that, in view of \cite[\textsection~II.1.11]{JacodShiryaev} together with the definition of $Y^{(T)}$ and the predictability of $g$,
\begin{align*}
\int_{(t,T]}g(t,u)\,\mu^p(\{t\}\times du)
&= \E\left[\int_{(t,T]}g(t,u)\,\mu(\{t\}\times du)\Bigr|\cF_{t-}\right]
=\E\bigl[\Delta Y^{(T)}_t\bigr|\cF_{t-}\bigr]	\\
&= \E\left[\int_{\mathbb{R}}y\,\mu^{(Y^{(T)},H)}(\{t\}\times dy\times\{0,1\})\Bigr|\cF_{t-}\right]	\\
&= \int_{\mathbb{R}}y\,\mu^{p,(Y^{(T)},H)}(\{t\}\times dy\times\{0,1\}).
\end{align*}
\end{proof}


\subsection{Proofs of the results of Section~\ref{sec:examples}}	\label{sec:proofs_cor}

\begin{proof}[Proof of Lemma \ref{lem:no_news_default}]
By definition of the process $Y^{(T)}$, the random set $\{\Delta Y^{(T)}\neq0\}$ is a subset of $\{(\omega,t)\in\Omega\times[0,T]:\mu(\omega;\{t\}\times[0,\T])>0\}$. Hence, condition \eqref{eq:no_news_default} implies that, up to an evanescent set, $\Delta Y^{(T)}\Delta H=0$, so that $\Psi * \mu^{(Y^{(T)},H)}=\Psi^{(1)} * \mu^{Y^{(T)}}$, using the notation introduced in \eqref{decomp:W}.
\end{proof}

\begin{proof}[Proof of Corollary \ref{cor:int_valued}]
Observe first that
\[
\bar{\mu}_t 
= \mu \bigl([0,\T]\times[0,t]\bigr)
= \mu\bigl([0,t]\times[0,t]\bigr) 
= \sum_{n\in\mathbb{N}}\ind{\tau_n\leq t,\sigma_n\leq t}
= \sum_{n\in\mathbb{N}}\ind{\tau_n\leq t}
\]
and, consequently, the decomposition \eqref{eq:dec_pred_proc} reduces to $\bar{\mu}_t=\sum_{0<s\leq t}\Delta\bar{\mu}_s$. 
This implies that condition {\em (i)} of Theorem \ref{thm:dc} reduces to condition {\em (i)} of the present corollary.  
Condition {\em (ii)} of the corollary directly follows from condition {\em (ii)} of Theorem \ref{thm:dc}.
Note also that
\[
(g^{(T)} * \mu)_t
= \sum_{n\in\mathbb{N}}g\left(\sigma_n,\tau_n\right)\ind{\tau_n\leq T}\ind{\sigma_n\leq t},
\]
so that the continuous part $(g^{(T)} * \mu)^c$ is null. 
Moreover, by the definition of the process $Y^{(T)}$,
\[
\Delta Y^{(T)}_t
= \int_0^Tg(t,u)\mu(\{t\}\times du)
= \Ind_{D}(t)\ind{\gamma_t\leq T}g(t,\gamma_t),
\]
where $D=\bigcup_{n\in\N}\dbra{\sigma_n}$ and $\gamma$ is an optional process such that $\tau_n=\gamma_{\sigma_n}$, for all $n\in\N$,
so that
\[
(\Psi^{(1)} * \mu^{Y^{(T)}})_t
= \sum_{0<s\leq t}e^{g(s,s)\Delta\bar{\mu}_s}\bigl(e^{-\Delta Y^{(T)}_s}-1\bigr)	\\
= \int_0^t\int_0^Te^{g(s,s)\Delta\bar{\mu}_s}\bigl(e^{-g(s,u)}-1\bigr)\mu(ds,du).
\]
Condition {\em (iii)} then follows from Lemma \ref{lem:no_news_default}, which implies that $\Psi * \mu^{p,(Y^{(T)},H)} = \Psi^{(1)} * \mu^{p,Y^{(T)}}$. 
Conditions {\em (iv)}-{\em (v)} of the corollary follow from conditions {\em (iv)}-{\em (v)} of Theorem~\ref{thm:dc}.
\end{proof}

\begin{proof}[Proof of Corollary \ref{cor:GS_case1}]
Under the present assumptions, the random measure $\mu(ds,du)$ is an integer-valued random measure. Since $\Q(\tau=\sigma_i)=0$, for all $i=1,\ldots,N$, condition \eqref{eq:no_news_default} holds, so that the assumptions of Corollary \ref{cor:int_valued} are satisfied. The process $(\bar{\mu}_t)_{0\leq t\leq \T}$ is given by  
\[
\bar{\mu}_t = \mu\bigl([0,\T]\times[0,t]\bigr) = \sum_{i=1}^N\ind{\tau_i\leq t}, 
\qquad\text{ for all } 0\leq t\leq \T,
\] 
so that conditions \emph{(i)-(ii)} of the present corollary follow from conditions \emph{(i)}-\emph{(ii)} of Corollary \ref{cor:int_valued}.
Condition {\em (iii)} of the present corollary corresponds to condition {\em (iv)} of Corollary \ref{cor:int_valued}, noting that $J^{{\rm ac}}_t=1$ and $F(t;du)=\xi_t(du)$, for all $t\in[0,\T]$.
Finally, under the present assumptions, conditions {\em (iii)} and {\em (v)} of Corollary \ref{cor:int_valued} are always satisfied.
\end{proof} 

\begin{proof}[Proof of Corollary \ref{cor:Merton}]
Note first that Assumption \ref{ass:rnd_meas} is clearly satisfied under the present assumptions.
Moreover, the process $(\bar{\mu}_t)_{0\leq t\leq \T}$ is simply given by
\[
\bar{\mu}_t = \mu\bigl([0,\T]\times[0,t]\bigr) 
= \mu\bigl(\{0\}\times[0,t]\bigr)
= \sum_{i=1}^{N}\ind{u_i\leq t},
\qquad \text{ for all }0\leq t\leq \T.
\]
Conditions {\em (i)}-\emph{(ii)} then follow from conditions {\em (i)}-\emph{(ii)} of Theorem \ref{thm:dc}. Moreover, for all $0\leq t\leq T\leq \T$, it holds that
\[
Y^{(T)}_t = (g^{(T)} * \mu)_t
= \int_0^Tg(0,u)\mu\bigl(\{0\}\times du\bigr)
= \sum_{i=1}^{N}g(0,u_i)\ind{u_i\leq T},
\] 
so that $Y^{(T)}_t=Y^{(T)}_0$, for all $0\leq t\leq T\leq \T$. 
In particular, $\Delta Y^{(T)}=0$, so that conditions {\em (iii)} and {\em (v)} of Theorem \ref{thm:dc} are automatically satisfied, since $\lambda=0$.
Condition {\em (iii)} of the corollary then immediately follows from condition {\em (iv)} of Theorem \ref{thm:dc}.
\end{proof}

\subsection{Proof of Theorem \ref{thm:dc2}} \label{sec:proof_recovery}

\begin{proof}[Proof of Theorem \ref{thm:dc2}]
In view of Lemma \ref{lem:compute_G}, credit risky bond prices admit the representation $P(t,T)=\cE(R)_t\,\exp(X^{(T)}_t)$, for all $0\leq t\leq T\leq \T$, where the process $X^{(T)}$ is defined as in \eqref{eq:XtT}. 
By the same arguments of Lemma \ref{lem:stoch_exp}, we have that
\[
P(t,T) = \cE\bigl(\widetilde{X}^{(T)}+R+[\widetilde{X}^{(T)},R]\bigr)_t.
\]
Note that $[\widetilde{X}^{(T)},R]=\sum_{0<s\leq \cdot}\Delta\widetilde{X}^{(T)}_s\Delta R_s$, since $R$ is of finite variation. 
Moreover, arguing similarly as in \eqref{eq:733}, it holds that
\begin{align*}
\Delta\widetilde{X}^{(T)}_t(1+\Delta R_t)
&= \bigl(e^{\Delta X^{(T)}_t}-1\bigr)(1+\Delta R_t)	\\
&= e^{g(t,t)\Delta\bar{\mu}_t}\left(e^{-\int_t^Tg(t,u)\mu(\{t\}\times du)}-1\right)\biggl(1+\int_{[-1,0]}x\mu^R(\{t\}\times dx)\biggr)	\\
&\quad
+ \left(e^{g(t,t)\Delta\bar{\mu}_t}-1\right)\biggl(1+\int_{[-1,0]}x\mu^R(\{t\}\times dx)\biggr).
\end{align*}
Hence,
\begin{align*}
\sum_{0<s\leq t}\Delta\widetilde{X}^{(T)}_s(1+\Delta R_s)
&= \bigl(\Psi * \mu^{(Y^{(T)},-R)}\bigr)_t
+ \sum_{0<s\leq t}\left(e^{g(s,s)\Delta\bar{\mu}_s}-1\right)\biggl(1+\int_{[-1,0]}x\mu^R(\{s\}\times dx)\biggr).
\end{align*}
Note that all the terms appearing in the last expression are a.s. finite, by the same arguments used after equation \eqref{eq:ineq_jumps} together with the fact that the process $R$ has bounded jumps.
Moreover,
\begin{align*}
&\sum_{0<s\leq t}\left(e^{g(s,s)\Delta\bar{\mu}_s}-1\right)\int_{[-1,0]}x\mu^R(\{s\}\times dx)	\\
&\qquad\qquad = (\text{local martingale})_t + \sum_{0<s\leq t}\left(e^{g(s,s)\Delta\bar{\mu}_s}-1\right)\int_{[-1,0]}x\mu^{p,R}(\{s\}\times dx).
\end{align*}
Hence, similarly as in the proof of Theorem \ref{thm:dc}, we obtain that
\begin{align*}
\widetilde{X}^{(T)}_t + R_t + [\widetilde{X}^{(T)},R]_t
&= (\text{local martingale})_t	\\
&\quad +\int_0^tf(s,s)ds
-\int_0^t\bar{a}(s,T)ds
-\int_0^t\bar{\alpha}(s,T)ds	
+\frac{1}{2}\int_0^t\bigl\|\bar{b}(s,T)+\bar{\beta}(s,T)\bigr\|^2ds	\\
&\quad
+\int_0^tg(s,s)m_sds
+\int_0^tg(s,s)d\nu_s
-(g^{(T)} * \mu)^c_t	\\
&\quad 
-C_t
+ \bigl(\Psi * \mu^{p,(Y^{(T)},-R)}\bigr)_t	\\
&\quad
+ \sum_{0<s\leq t}\left(e^{g(s,s)\Delta\bar{\mu}_s}-1\right)\biggl(1+\int_{[-1,0]}x\mu^{p,R}(\{s\}\times dx)\biggr).
\end{align*}
Conditions {\em (i)}-{\em (v)} then follow by a similar analysis as in the proof of Theorem \ref{thm:dc}.
\end{proof}

\appendix

\section{A variant of the stochastic Fubini theorem}

In this appendix, we show that Assumptions \ref{ass:rnd_meas} and \ref{ass:processes} imply that it is possible to interchange the order of integration in the term (3) appearing in equation \eqref{eq:dyn_3} in the proof of Lemma \ref{lem:compute_G}. 

As a preliminary, following \cite[Exercise 3.6]{Jacod}, we write the following unique decomposition of $\mu$:
\[
\mu(\omega;ds,du) = \mu^{(c)}(\omega;ds,du) + \mu^{(d)}(\omega;ds,du)
=\mu^{(c)}(\omega;ds,du) + 
\sum_{k=1}^{+\infty}
\ind{T_k(\omega)\leq\T}\delta_{T_k(\omega)}(ds)F_k(\omega;du),
\]
where $\mu^{(c)}$ is a random measure satisfying $\mu^{(c)}(\{t\}\times [0,\T])=0$ for all $t\in[0,\T]$, $(T_k)_{k\in\N}$ is a family of disjoint stopping times and, for every $k\in\N$, $F_k(\omega;du)$ is a kernel from $(\Omega,\cF_{T_k})$ into $([0,\T],\mathcal{B}([0,\T]))$ satisfying $\sup_{\omega\in\Omega}F_k(\omega;[0,\T])<+\infty$, for all $k\in\N$.
Furthermore, by the same arguments used in part (d) of the proof of \cite[Theorem II.1.8]{JacodShiryaev}, we can write
\[
\mu^{(c)}(\omega;ds,du) = K(\omega,s;du)dA_s(\omega),
\]
where $(A_t)_{0\leq t\leq \T}$ is an integrable increasing predictable process and $K(\omega,s;du)$ is a kernel from $(\Omega\times[0,\T],\mathcal{P})$ into $([0,\T],\mathcal{B}([0,\T]))$. Moreover, since $\mu^{(c)}(\{t\}\times [0,\T])=0$ for all $t\in[0,\T]$, the process $(A_t)_{0\leq t\leq \T}$ is continuous.
Summing up, we get the general decomposition
\begin{equation}	\label{eq:decomposition_mu}
\mu(\omega;ds,du) 
= K(\omega,s;du)dA_s(\omega) + \sum_{k=1}^{+\infty}\ind{T_k(\omega)\leq\T}\delta_{T_k(\omega)}(ds)F_k(\omega;du).
\end{equation}

\begin{lemma}	\label{lem:stoch_Fubini}
Suppose that Assumptions \ref{ass:rnd_meas} and \ref{ass:processes} hold. Then, for all $0\leq t\leq T\leq \T$, it holds that
\[
\int_0^t\int_0^T \int_s^t\Ind_{[0,u)}(v)\beta(v,u)dW_v\, \mu(ds,du)
=  \int_0^t \int_0^v \int_0^T\Ind_{[0,u)}(v)\beta(v,u) \mu(ds,du)\, dW_v.
\]
\end{lemma}
\begin{proof}
We first consider the integral with respect to the purely discontinuous part $\mu^{(d)}$ of $\mu$:
\[
\int_0^t\int_0^T\int_s^t\Ind_{[0,u)}(v)\beta(v,u)dW_v\,\mu^{(d)}(ds,du)
= \sum_{k=1}^{+\infty}\ind{T_k\leq t}\int_0^T\int_{T_k}^t\Ind_{[0,u)}(v)\beta(v,u)dW_v\,F_k(du).
\]
For each $k\in\N$, let us define the filtration $\FF^k=(\cF^k_t)_{0\leq t\leq\T}$ by $\cF^k_t:=\cF_{(T_k+t)\wedge \T}$, for all $t\in[0,\T]$, and the stochastic process $W^k=(W^k_t)_{0\leq t\leq \T}$ by $W^k_t:=W_{(T_k+t)\wedge\T}-W_{T_k\wedge\T}$, for all $t\in[0,\T]$. The strong Markov property implies that $W^k$ is a Brownian motion in the filtration $\FF^k$, for each $k\in\N$. Therefore, for each $k\in\N$, we can write:
\begin{align*}
\ind{T_k\leq t}\int_0^T\int_{T_k}^t\Ind_{[0,u)}(v)\beta(v,u)dW_v\,F_k(du)
&= \ind{T_k\leq t}\int_0^T\int_{0}^{t-T_k}h^k(v,u)dW^k_v\,F_k(du),
\end{align*}
where the process $h^k(\cdot,u):=\Ind_{[0,u)}(\cdot+T_k)\beta(\cdot+T_k,u)$ is $\FF^k$-adapted and $\int_0^{t-T_k}h^k(v,u)dW^k_v$ makes sense as a Brownian stochastic integral parametrised by $u$ in the filtration $\FF^k$. Since the kernel $F_k(du)$ is $\cF^k_0$-measurable, we can apply the stochastic Fubini theorem of \cite[Theorem IV.65]{Protter}\footnote{
A careful inspection of the proof of \cite[Theorem IV.65]{Protter} reveals that the stochastic Fubini theorem holds true even if the measure $F_k(du)$ is not deterministic but only $\cF^k_0$-measurable, since $\sup_{\omega\in\Omega}F_k(\omega;[0,\T])<+\infty$.
The integrability condition $\int_0^{T-T_k}\int_0^T(h^{k,i}(v,u))^2F_k(du)dv<+\infty$ a.s., for all $i=1,\ldots,n$, required for the application of \cite[Theorem IV.65]{Protter}, is implied by the requirement $\int_0^{\T}\int_0^{\T}\|\beta(v,u)\|^2\mu_v(du)dv<+\infty$ a.s. appearing in Assumption \ref{ass:processes}.} in the filtration $\FF^k$ to obtain
\begin{align*}
\ind{T_k\leq t}\int_0^T\int_{0}^{t-T_k}h^k(v,u)dW^k_v\,F_k(du)
&= \ind{T_k\leq t}\int_{0}^{t-T_k}\int_0^Th^k(v,u)F_k(du)\,dW^k_v	\\
&= \ind{T_k\leq t}\int_{0}^{t}\ind{T_k\leq v}\int_0^T\Ind_{[0,u)}(v)\beta(v,u)F_k(du)\,dW_v\\
&= \int_{0}^{t}\int_0^T\ind{T_k\leq v}\Ind_{[0,u)}(v)\beta(v,u)F_k(du)\,dW_v.
\end{align*}
Applying this argument to every $k\in\N$, we get that
\begin{align*}
\int_0^t\int_0^T\int_s^t\Ind_{[0,u)}(v)\beta(v,u)dW_v\,\mu^{(d)}(ds,du)
&= \sum_{k=1}^{+\infty}\int_{0}^{t}\int_0^T\ind{T_k\leq v}\Ind_{[0,u)}(v)\beta(v,u)F_k(du)\,dW_v	\\
&=  \lim_{N\rightarrow+\infty}\int_{0}^{t}\ell^N(v,T) dW_v,
\end{align*}
where $\ell^N(\cdot,T):=\sum_{k=1}^{N}\int_0^T\ind{T_k\leq \cdot}\Ind_{[0,u)}(\cdot)\beta(\cdot,u)F_k(du)$ is an $\FF$-adapted process, integrable with respect to the Brownian motion $W$. 
For each fixed $v$ and $T$, the sequence $(\ell^N(v,T))_{N\in\N}$ converges pointwise to
\begin{align*}
\sum_{k=1}^{+\infty}\int_0^T\ind{T_k\leq v}\Ind_{[0,u)}(v)\beta(v,u)F_k(du)
&=\int_0^{v}\int_{(v,T]}\beta(v,u)\mu^{(d)}(ds,du)
=: \bar{\beta}^{(d)}(v,T).
\end{align*}
By Assumption \ref{ass:processes}, the process $\bar{\beta}^{(d)}(\cdot,T)$ is integrable with respect to $W$ (this follows exactly as in the first part of the proof of Lemma 4.1).
Moreover, it can be checked that 
\[
|\ell^{N,i}(v,T)| \vee  |\bar{\beta}^{(d),i}(v,T)|
\leq \int_{(v,T]}|\beta^i(v,u)|\mu^{(d)}([0,v]\times du),
\]
for each $i=1,\ldots,n$, and the process $(\int_{(v,T]}|\beta^i(v,u)|\mu^{(d)}([0,v]\times du))_{0\leq v\leq \T}$ is integrable with respect to $W^i$, for each $i=1,\ldots,n$. 
The dominated convergence theorem for stochastic integrals (see \cite[Theorem IV.32]{Protter}) implies that $\int_0^{\cdot}\ell^N(v,T)dW_v$ converges uniformly on compacts in probability to $\int_0^{\cdot}\bar{\beta}^{(d)}(v,T)dW_v$ as $N\rightarrow+\infty$. Putting together the above results, we have shown that
\begin{align*}
\int_0^t\int_0^T\int_s^t\Ind_{[0,u)}(v)\beta(v,u)dW_v\,\mu^{(d)}(ds,du)
&= \lim_{N\rightarrow+\infty}\int_{0}^{t}\ell^N(v,T) dW_v\\
&= \int_0^t\bar{\beta}^{(d)}(v,T)dW_v	\\
&= \int_0^t\int_0^v\int_0^T\Ind_{[0,u)}(v)\beta(v,u)\mu^{(d)}(ds,du)\,dW_v.
\end{align*}

It remains to perform an analogous interchange of the order of integration with respect to the continuous random measure $\mu^{(c)}$ in the term
\begin{equation}	\label{eq:cont_part}
\int_0^t\int_0^T\int_s^t\Ind_{[0,u)}(v)\beta(v,u)dW_v\,\mu^{(c)}(ds,du)
= \int_0^t\int_0^T\int_s^t\Ind_{[0,u)}(v)\beta(v,u)dW_v\,K(s;du)dA_s.
\end{equation}
As a preliminary, we can assume without loss of generality that the process $(A_t)_{0\leq t\leq\T}$ is strictly increasing. Indeed, if $\mu^{(c)}(ds,du)=K(s;du)dA_s$ is an arbitrary disintegration of the random measure $\mu^{(c)}$ into a predictable kernel $K$ and an increasing continuous process $A$, let us consider the unique decomposition (compare with the proof of Lemma 2.1)
\[
A_t
=\int_0^ta_sds + \int_0^t\Ind_N(s)dA_s
=:\int_0^ta_sds+A^{{\rm sg}}_t,
\] 
where $(a_t)_{0\leq t\leq \T}$ is a predictable non-negative integrable process and $N$ is a predictable set such that its sections have Lebesgue measure zero a.s. Let $A_t':=t+A^{{\rm sg}}_t$, for all $t\in[0,\T]$, and  $K'(t;du):=(a_t\Ind_{N^c}(t)+\Ind_{N}(t))K(t;du)$. It is clear that the process $(A'_t)_{0\leq t\leq \T}$ is continuous and strictly increasing and  $K'(t;du)$ is a kernel from $(\Omega\times[0,\T],\mathcal{P})$ into $([0,\T],\mathcal{B}([0,\T]))$.
Moreover, it holds that $K'(t;du)dA'_t=K(t;du)dA_t$.
We can furthermore assume that $K(s;[0,\T])\leq1$  for all $s\in[0,\T]$. Indeed, if $\mu^{(c)}(ds,du)=K(s;du)dA_s$ is a disintegration of the random measure $\mu^{(c)}$ with respect to a continuous and strictly increasing process $(A_t)_{0\leq t\leq \T}$, let
\[
\widetilde{K}(s;du) := \frac{K(s;du)}{K(s;[0,\T])+\epsilon}
\qquad\text{and}\qquad
\widetilde{A} := \int_0^{\cdot}\bigl(K(s;[0,\T])+\epsilon\bigr)dA_s,
\]
for some $\epsilon>0$.
Clearly, we have that $\widetilde{K}(s;[0,\T])\leq1$ for all $s\in[0,\T]$ and the process $(\widetilde{A}_t)_{0\leq t\leq \T}$ is continuous and strictly increasing  (since $(A_t)_{0\leq t\leq \T}$ is continuous and strictly increasing). Moreover, it holds that 
$
\widetilde{K}(s;du)d\widetilde{A}_s
= K(s;du)dA_s = \mu^{(c)}(ds,du)
$.
Summing up, in the disintegration $\mu^{(c)}(ds,du)=K(s;du)dA_s$ we can always assume that the kernel $K$ satisfies $K(s;[0,\T])\leq 1$ for all $s\in[0,\T]$ and that the process $(A_t)_{0\leq t\leq \T}$ is continuous and strictly increasing.
In the following, we will make use of these two properties.

Consider the inner integral in \eqref{eq:cont_part}: $\int_0^T\int_s^t\Ind_{[0,u)}(v)\beta(v,u)dW_v\,K(s;du)$, for a fixed $s\in[0,t]$. By the same arguments used in the first part of the proof, 
we can use the stochastic Fubini theorem (see \cite[Theorem IV.65]{Protter}\footnote{The integrability condition $\int_s^t\int_0^T\Ind_{[0,u)}(v)(\beta^{i}(v,u))^2K(s;du)dv<+\infty$ a.s., for each $0\leq s\leq t\leq T\leq\T$ and $i=1,\ldots,n$, required for the application of \cite[Theorem IV.65]{Protter} is implied by the stronger condition $\int_0^{\T}\int_0^{\T}\|\beta(v,u)\|^2\mu_v(du)dv<+\infty$ a.s. appearing in Assumption \ref{ass:processes}.}) to deduce that, for all $0\leq s\leq t\leq T\leq \T$,
\[
\int_0^T\int_s^t\Ind_{[0,u)}(v)\beta(v,u)dW_v\,K(s;du)
= \int_s^t\int_0^T\Ind_{[0,u)}(v)\beta(v,u)K(s;du)\,dW_v.
\]
Hence:
\[
\int_0^t\int_0^T\int_s^t\Ind_{[0,u)}(v)\beta(v,u)dW_v\,\mu^{(c)}(ds,du)
= \int_0^tX^{t,T}_s\,dA_s,
\]
where the continuous process $(X^{t,T}_s)_{0\leq s\leq t}$ is defined by $X^{t,T}_s:=\int_s^t\int_{(v,T]}\beta(v,u)K(s;du)\,dW_v$, for all $s\in[0,t]$. Since the process $(A_t)_{0\leq t\leq \T}$ is continuous and strictly increasing, the change of time process $(C_t)_{t\geq0}$ defined by $C_t:=\inf\{s\in[0,\T]:A_s\geq t\}$, for $t\in\R_+$, is continuous and strictly increasing, so that $A_{C_t}=t$ and $C_{A_t}=t$, for all $t\in[0,\T]$. Moreover, arguing similarly as in the proof of \cite[Lemma I.3.12]{JacodShiryaev}, we have that
\[
\int_0^tX^{t,T}_s\,dA_s 
= \int_0^{\infty}\ind{C_s\leq t}X^{t,T}_{C_s}ds,
\]
meaning that
\begin{align*}
\int_0^t\int_0^T\int_s^t\Ind_{[0,u)}(v)\beta(v,u)dW_v\,\mu^{(c)}(ds,du)
&= \int_0^{\infty}\ind{C_s\leq t}\int_{C_s}^t\int_{(v,T]}\beta(v,u)K(C_s;du)\,dW_v\,ds\\
&= \int_0^{\infty}\int_{0}^t\ind{C_s\leq v}\int_{(v,T]}\beta(v,u)K(C_s;du)\,dW_v\,ds\\
&= \int_0^{\infty}\int_{0}^tL(v,s)\,dW_v\,ds,
\end{align*}
where the process $(L^{(T)}(t,s))_{0\leq t\leq T}$ is defined by $L^{(T)}(t,s):=\ind{C_s\leq t}\int_{(t,T]}\beta(t,u)K(C_s;du)$, for every $t\in[0,T]$ and  $s\in\R_+$. Note that, since each $C_s$ is a predictable time (due to the continuity of $A$) and the kernel $K(C_s;du)$ if $\cF_{C_s-}$-measurable, the process $(L^{(T)}(t,s))_{0\leq t\leq T}$ is predictable.
At this point, the stochastic Fubini theorem of \cite[Theorem IV.65]{Protter} implies that
\begin{equation}	\label{eq:last_Fubini}
\int_0^{\infty}\int_{0}^tL^{(T)}(v,s)\,dW_v\,ds
= \int_{0}^t\int_0^{\infty}L^{(T)}(v,s)\,ds\,dW_v.
\end{equation}
The integrability requirement appearing in \cite[Theorem IV.65]{Protter} is satisfied. Indeed, using H\"older's inequality and recalling that $K(s;[0,\T])\leq1$ for all $s\in[0,\T]$, it holds that
\begin{align*}
\int_0^{\T}\int_0^{\infty}(L^{(T),i}(v,s))^2ds\,dv
&= \int_0^{\T}\int_0^{\infty}\biggl(\ind{C_s\leq v}\int_{(v,T]}\beta^i(v,u)K(C_s;du)\biggr)^2ds\,dv	\\
&\leq  \int_0^{\T}\int_0^{\infty}\ind{C_s\leq v}\int_{(v,T]}(\beta^i(v,u))^2K(C_s;du)ds\,dv	\\
&= \int_0^{\T}\int_0^v\int_{(v,T]}(\beta^i(v,u))^2K(s;du)dA_s\,dv\\
&= \int_0^{\T}\int_{(v,T]}(\beta^i(v,u))^2\mu^{(c)}([0,v]\times du)\,dv\\
&\leq \int_0^{\T}\int_0^{\T}\|\beta(v,u)\|^2\mu_v(du)dv < +\infty,
\end{align*}
for every $i=1,\ldots,n$.
Hence, continuing from \eqref{eq:last_Fubini}, we get that
\begin{align*}
\int_0^t\int_0^T\int_s^t\Ind_{[0,u)}(v)\beta(v,u)dW_v\,\mu^{(c)}(ds,du)
&= \int_{0}^t\int_0^{\infty}L^{(T)}(v,s)\,ds\,dW_v	\\
&= \int_{0}^t\int_0^{\infty}\ind{C_s\leq v}\int_{(v,T]}\beta(v,u)K(C_s;du)\,ds\,dW_v	\\
&= \int_{0}^t\int_0^v\int_{(v,T]}\beta(v,u)K(s;du)\,dA_s\,dW_v	\\
&= \int_{0}^t\int_0^v\int_0^T\Ind_{[0,u)}(v)\beta(v,u)\mu^{(c)}(ds,du)\,dW_v.
\end{align*}

Finally, by linearity of the integral, it holds that
\begin{align*}
&\int_0^t\int_0^T\int_s^t\Ind_{[0,u)}(v)\beta(v,u)dW_v\,\mu(ds,du)	\\
&= \int_0^t\int_0^T\int_s^t\Ind_{[0,u)}(v)\beta(v,u)dW_v\,\mu^{(c)}(ds,du)
+ \int_0^t\int_0^T\int_s^t\Ind_{[0,u)}(v)\beta(v,u)dW_v\,\mu^{(d)}(ds,du)	\\
&=  \int_{0}^t\int_0^v\int_0^T\Ind_{[0,u)}(v)\beta(v,u)\mu^{(c)}(ds,du)\,dW_v
+  \int_{0}^t\int_0^v\int_0^T\Ind_{[0,u)}(v)\beta(v,u)\mu^{(d)}(ds,du)\,dW_v	\\
&=  \int_{0}^t\biggl(\int_0^v\int_0^T\Ind_{[0,u)}(v)\beta(v,u)\mu^{(c)}(ds,du)+\int_0^v\int_0^T\Ind_{[0,u)}(v)\beta(v,u)\mu^{(d)}(ds,du)\biggr)dW_v\\
&= \int_{0}^t\int_0^v\int_0^T\Ind_{[0,u)}(v)\beta(v,u)\mu(ds,du)\,dW_v,
\end{align*}
thus completing the proof.
\end{proof}


\end{document}